\documentclass[10pt, conference, letterpaper]{IEEEtran}
\usepackage{tikz}
\usepackage{url}
\usepackage{pdfpages}
\usepackage{afterpage}
\usepackage{color, colortbl}
\usepackage{graphicx}
\usepackage{amsmath}
\usepackage{amsfonts}
\usepackage{import}
\usepackage{pifont}
\usepackage{multirow}
\usepackage{subfigure}
\usepackage{footmisc}
\usepackage{algorithm}
\usepackage[noend]{algpseudocode}
\usepackage{amssymb}
\usepackage{amsthm}
\usepackage{bm}
\theoremstyle{definition}
\newtheorem{definition}{Definition}[]

\newtheorem{condition}{Conditions}[]
\newtheorem{theorem}{Theorem}[]
\usepackage[
separate-uncertainty = true,
multi-part-units = repeat
]{siunitx}

\def\BState{\State\hskip-\ALG@thistlm}

\definecolor{Gray}{gray}{0.9}
\definecolor{LightCyan}{rgb}{0.88,1,1}
\pdfoutput=1
\graphicspath {{figures/}}

\makeatletter
\DeclareRobustCommand*\textsubscript[1]{%
	\@textsubscript{\selectfont#1}}
\def\@textsubscript#1{%
	{\m@th\ensuremath{_{\mbox{\fontsize\sf@size\z@#1}}}}}
\makeatother

\usepackage[flushleft]{threeparttable}

\begin{document}
	
	\title{MG-WFBP: Efficient Data Communication for Distributed Synchronous SGD Algorithms}
	
\author{\IEEEauthorblockN{Shaohuai Shi\IEEEauthorrefmark{1}, Xiaowen Chu\IEEEauthorrefmark{1}, Bo Li\IEEEauthorrefmark{2}}
	\IEEEauthorblockA{\IEEEauthorrefmark{1}Department of Computer Science, Hong Kong Baptist University
		\\\IEEEauthorrefmark{2}Department of Computer Science and Engineering, The Hong Kong University of Science and Technology
		\\\IEEEauthorrefmark{1}\{csshshi, chxw\}@comp.hkbu.edu.hk, \IEEEauthorrefmark{2}bli@cse.ust.hk}
	}
	
	\maketitle
	
\begin{abstract}
Distributed synchronous stochastic gradient descent has been widely used to train deep neural networks on computer clusters. With the increase of computational power, network communications have become one limiting factor on system scalability. In this paper, we observe that many deep neural networks have a large number of layers with only a small amount of data to be communicated. Based on the fact that merging some short communication tasks into a single one may reduce the overall communication time, we formulate an optimization problem to minimize the training iteration time. We develop an optimal solution named merged-gradient WFBP (MG-WFBP) and implement it in our open-source deep learning platform B-Caffe. Our experimental results on an 8-node GPU cluster with 10GbE interconnect and trace-based simulation results on a 64-node cluster both show that the MG-WFBP algorithm can achieve much better scaling efficiency than existing methods WFBP and SyncEASGD.
\end{abstract}

\begin{IEEEkeywords}
	Deep Learning; GPU; Distributed Stochastic Gradient Descent; Gradient Communication; Merged-gradient
\end{IEEEkeywords}

\section{Introduction}
The data-parallel synchronous stochastic gradient descent (S-SGD) method is commonly used as the optimizer to train the large scale deep neural networks (DNNs) \cite{dean2012large}\cite{goyal2017accurate}. In S-SGD, the computing tasks for each mini-batch of training data is distributed to a cluster of computing nodes, and the individual results are aggregated to update the global network model before the next iteration can begin. However, with more computing nodes and the fast-growing computing power of hardware accelerators, the data communication between computing nodes gradually becomes the performance bottleneck \cite{watcharapichat2016ako}\cite{cui2016geeps}. For example, the computing power of Nvidia GPUs has increased by 30x in the last 10 years, whilst it took about 15 years for the network speed to improve from 10Gbps to 100Gbps. Hence it becomes a critical issue to address the imbalance between computing and communication.

Some recent work try to reduce the impact of data communication at both algorithmic and system levels. On one hand, gradients would be quantized and compressed \cite{alistarh2017qsgd}\cite{lin2017deep}\cite{wen2017terngrad} in order to reduce the data size during communication so that the overhead of communication could be reduced, but these methods usually sacrifice the model accuracy. On the other hand, the HPC community has proposed several methods to increase the communication performance of the cluster using both hardware and software approaches \cite{potluri2013efficient}.

In terms of hardware, InfiniBand (IB) and Omni-Path networks can provide much higher communication bandwidth, and are deployed to reduce the performance gap between communication and computation \cite{bayatpour2017scalable}. Regarding the software, the implementation of message passing interface (MPI) has been further optimized to support efficient communication in DNN trainings \cite{bayatpour2017scalable}\cite{awan2017s}. The scaling efficiency of distributed deep learning systems can be modeled as a function of communication-to-computation ratio \cite{wen2017terngrad}. For example, training ResNet-50 \cite{he2016deep} requires about 7.8 billion floating point operations in computation, while it takes 102 MB data communication in one iteration. Higher communication-to-computation ratio results in lower scaling efficiency. 

The layered structure of DNNs makes it possible to overlap the communication and computation during the backward propagation \cite{awan2017s}\cite{zhang2017poseidon}, which is known as wait-free backpropagation (WFBP). WFBP begins to exchange the gradients of a layer immediately after they have been calculated; so if the data communication time of a layer is shorter than the computation time of the gradients of its previous layer, then this communication cost can be fully hidden. However, if very fast hardware accelerators are used while the network speed is relatively slow (i.e., a high communication-to-computation ratio), there can exist many layers whose communication time is longer than the corresponding computation time. In such cases, it becomes important to optimize the communications. We observe that the layer-wise gradient communication in WFBP is suboptimal due to the fact that transmitting a small amount of data cannot fully utilize the network bandwidth in current network topologies due to the startup time of message transmitting. Even the RDMA-based network is difficult to eliminate the high overhead of startup time when transmitting messages \cite{handley2017re}\cite{guo2016rdma}. For example, on our 10GbE platform, exchanging a 200 KB vector across 8 nodes using MPI requires about 1.5ms, while exchanging a 400 KB vector only requires 1.8ms, which means we can merge two 200 KB vectors to one 400 KB vector to reduce the total communication time. Yang et al., \cite{you2017scaling} have also recently noticed this problem, and propose a single-layer communication (SyncEASGD) method in which all the gradients are merged together and transferred once per iteration. As compared to the layer-wise communication in WFBP, it can reduce most of the startup time of data communications. But in their proposed method, gradient communication can only start after the backward pass for all layers are finished, thus they miss the opportunity to overlap the communication with computation. 

We argue that the best way to reduce the training time needs to consider not only how to overlap communication with computation, but also how to improve the communication efficiency by avoiding transmitting small amount of data.

In this paper, we first formulate the communication scheduling problem in S-SGD as an optimization problem that aims to minimize the total training time of an iteration. And then we propose a merged-gradient wait-free backward propagation (MG-WFBP) method and prove its optimality. The time complexity of MG-WFBP is $O(L^2)$ where $L$ is the number of layers in the DNN, and it only needs to be executed once before the whole training process. We implement MG-WFBP in our open-source distributed DL training platform B-Caffe\footnote{https://github.com/shyhuai/B-Caffe.}, and evaluate its performance using two popular DNNs (i.e., GoogleNet \cite{szegedy2015going} and ResNet-50 \cite{he2016deep}). The experimental results on an 8-node cluster with Nvidia Tesla K80 GPU and 10GbE show that MG-WFBP can achieve about $1.2$x to $1.36$x improvement than the state-of-the-art communication algorithms WFBP and SyncEASGD, respectively. To investigate its performance on large clusters, we resolve to trace-based simulation (due to limited hardware resources) on a 64-node cluster. In the 64-node simulation, the results show that MG-WFBP performs more than $1.7$x and $1.3$x speedups compared to WFBP and SyncEASGD respectively. Our contributions are summarized as follows:
\begin{itemize}
	\item We formulate an optimization problem for minimizing the training time of DNNs by merging consecutive data communications, and propose an optimal solution with very low computational cost.
	\item We implement our MG-WFBP algorithm in B-Caffe and make it open-source.
    \item We evaluate the performance of MG-WFBP through both real experiments and simulations, and make comparisons with WFBP, SyncEASGD and TensorFlow.
\end{itemize}

The rest of the paper is organized as follows. We present the preliminaries in Section \ref{s:pre}, followed by the formulation of the existing problem in Section \ref{s:profor}. We derive an optimal solution to the problem and then present our MG-WFBP S-SGD algorithm in Section \ref{s:method}. Section \ref{s:eval} demonstrates the evaluation of the proposed method with experimental results. Section \ref{s:bm} introduces the related work, and finally we conclude this paper in Section \ref{s:conclusion}.

\section{Preliminaries}\label{s:pre}

For ease of presentation, we summarize the frequently used mathematical notations in Table \ref{table:notation}.
\begin{table}[!ht]
	\centering
	\caption{Frequently used notations}
	\label{table:notation}
	\begin{tabular}{|l|l|}
		\hline
		Name &  Description \\\cline{1-2}
		\hline
		\hline
		$N$ & The number of nodes in the cluster. \\
		$\alpha$ & Latency (startup time) of the network between two nodes. \\
		$\beta$ & Transmission time per byte between two nodes. \\
		$\gamma$ & Summation time of two floating point numbers in one node. \\
		$a$ & Latency (startup time) of all-reduce.\\
		$b$ & Transmission and computation time per byte of all-reduce. \\
		$M$ & The size of a message in bytes. \\\cline{1-2}
		$W$ & Weights of the DNN. \\		
		$D_i^g$ & The input data size for the $g^{th}$ node at the $i^{th}$ mini-batch.\\\cline{1-2}
		$L$ & The number of learnable layers of a DNN.\\
		$p^{(l)}$ & The number of parameters in the learnable layer $l$.\\
		$t_{iter}$ & Time of an iteration.\\
		$t_{f}$ & Time of the forward pass in each iteration.\\
		$t_{b}$ & Time of the backward propagation in each iteration.\\
		$t_{u}$ & Time of the model update in each iteration.\\
		$t_{b}^{(l)}$ & Time of the backward propagation of layer $l$ in each iteration.\\
		$\tau_{b}^{(l)}$ & The timestamp when layer $l$ begins calculating gradients.\\
		$\mu_{b}^{(l)}$ & The timestamp when layer $l$ finishes calculating gradients.\\
		$t_{c}$ & Time of gradient aggregation in each iteration.\\
		$t_{c}^{no}$ & The not overlapped communication cost in each iteration.\\
		$t_{c}^{(l)}$ & Time of gradient aggregation of layer $l$ in each iteration.\\
		$\tau_{c}^{(l)}$ & The timestamp when layer $l$ begins communicating gradients.\\
		$\mu_{c}^{(l)}$ & The timestamp when layer $l$ finishes communicating gradients.\\\cline{1-2}
	\end{tabular}
\end{table}
\subsection{Mini-batch SGD}

The DNN needs a loss function $\mathcal{L}(W,D)$, where $W$ and $D$ are the model weights and the input data respectively, to define the differences between the prediction values and the ground truth. To minimize the loss function, the mini-batch SGD updates the parameters iteratively. Typically, the $i^{th}$ iteration of the training includes four steps: 1) A mini-batch of data $D_i$ ($D_i\subset D$) is read as inputs of the DNN. 2) $D_i$ is fed forward across the neural network from layer $1$ to layer $L$ to compute the prediction values at the last layer, and the value of the loss function $\mathcal{L}(W,D)$ is computed. 3) The first order gradients w.r.t. parameters and inputs are calculated and backpropagated with from layer $L$ to layer $1$. 4) Finally, the parameters are updated with the layer-wise gradients. The training is terminated when some stopping criteria are matched. The update of $W$ can be formulated as follows:
\begin{equation}
W_{i+1}=W_{i}-\eta\cdot\nabla\mathcal{L}(W_{i},D_{i}),
\end{equation}
where $\eta$ is the learning rate of SGD, $W_{i}$ is the weights at $i^{th}$ iteration, and $\nabla\mathcal{L}(W_{i},D_{i})$ are the gradients. The time consumed in the training processes are mainly in steps 2 and 3, because step 1 of the $i^{th}$ iteration can overlap with the $(i-1)^{th}$ iteration, and the time of step 4 is negligible. Therefore, we can simplify the time-line of SGD to forward and backward passes. The time of one iteration is represented by $t_{iter}=t_f+t_b$.

\subsection{Synchronized SGD for clusters}
For large-scale DNNs, the data-parallelism synchronized SGD (S-SGD) is widely applied to train models with multiple workers (say $N$ workers, and indexed by $g$). Each worker takes a different mini-batch of data $D_{i}^{g}$ and forwards it by step 2), and then follows step 3) to calculate the gradients $\nabla\mathcal{L}(W_{i},D_{i}^{g})$. In this way, each worker has a copy of the model, while the gradients are not the same in each iteration since the input data are different; therefore, to keep explicitly the same as SGD, it needs to average the gradients from different workers before updating the model. The update formula of parameters is rewritten as
\begin{equation}\label{equ:ssgd}
W_{i+1}=W_{i}-\eta\cdot\frac{1}{N}\sum_{g=1}^{N}\nabla\mathcal{L}(W_{i},D_{i}^{g}).
\end{equation}
As a result, the averaging operation of gradients across the cluster involves extra computation and communication overheads such that it is not easy to achieve linear scaling in the distributed SGD training. The time-line of the naive S-SGD (i.e., computation and communication are not overlapped) with communication overheads is illustrated in Fig. \ref{fig:ssgd}. The naive S-SGD algorithm suffers from the waiting period of data communication of model synchronization at every iteration. The iteration time of the naive S-SGD can be estimated as
\begin{equation}\label{equ:tssgd}
t_{iter}=t_{f}+t_{b}+t_{c},
\end{equation}
where $t_{b}=\sum_{l=1}^{L}t_{b}^{(l)}$ is the backward propagation time and $t_{c}=\sum_{l=1}^{L}t_{c}^{(l)}$ is the gradient aggregation time which heavily relies on the communication speed.
\begin{figure}[!ht]
	\centering
	\includegraphics[width=\linewidth]{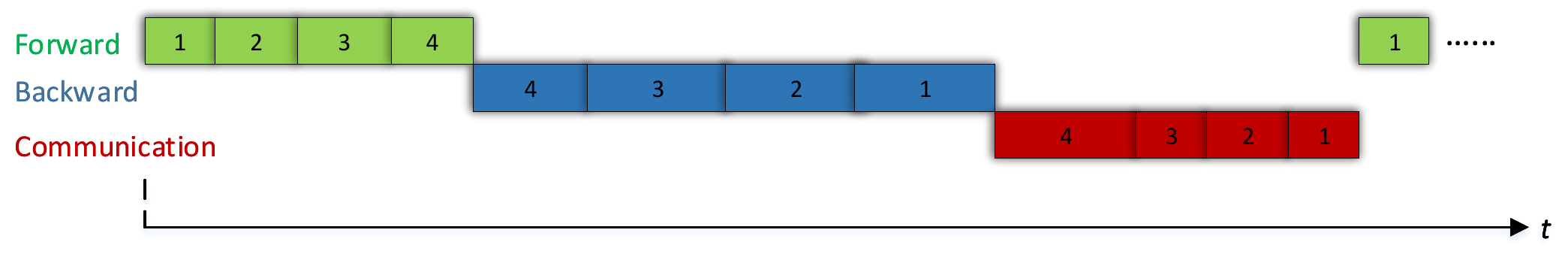}
	\vspace{-20pt}
	\caption{The time-line of naive S-SGD for a 4-layer network with communication overheads. The feed forward cannot be started until the end of model communication (red rectangles) of its previous iteration.}
	\label{fig:ssgd}
\end{figure}

Considering S-SGD using weak-scaling running on $N$ nodes, we define the speedup of S-SGD compared to the single-node SGD:
\begin{equation}
S(N)=\frac{N|D_i^{g}|/(t_f+t_b+t_c)}{|D_i^{g}|/(t_f+t_b)}=\frac{N}{1+\frac{t_c}{t_f+t_b}},
\end{equation}
where $|D_i^{g}|$ is the number of samples per node at the $i^{th}$ iteration. Let $r=\frac{t_c}{t_f+t_b}$, which reflects the communication-to-computation ratio, we have
\begin{equation}\label{equ:speedupsync}
S(N)=\frac{N}{1+r}.
\end{equation}

\subsection{WFBP}
In WFBP, the layer-wise gradient communication can be overlapped with the backward propagation of its previous layer. An example of S-SGD with WFBP is illustrated in Fig. \ref{fig:wfbp}. For simplicity, we assume that the start timestamp of the forward pass is $0$, then the start timestamp of backward pass of each layer can be represented by
\begin{equation}\label{equ:startcomp}
\tau_b^{(l)}=
\begin{cases}
t_f & l=L\\
\tau_b^{(l+1)}+t_b^{(l+1)} & 1\leq l<L
\end{cases}.
\end{equation}
And the start timestamp of communication of each layer can be represented by
\begin{equation}\label{equ:startt}
\tau_c^{(l)}=
\begin{cases}
\tau_b^{(l)}+t_b^{(l)} & l=L\\
\text{max}\{\tau_c^{(l+1)}+t_c^{(l+1)}, \tau_b^{(l)}+t_b^{(l)}\} & 1\leq l<L
\end{cases}.
\end{equation}
The iteration time can be rewritten as
\begin{equation}\label{equ:wfbpiter}
\begin{split}
t_{iter}&=t_f+t_b^{(L)}+t_c^{(1)}-\tau_c^{(L)}+\tau_c^{(1)}\\
%&=t_f+t_b^{(L)}+t_c^{(1)}-(t_f+t_b^{(L)})+\tau_c^{(1)}\\
&=t_c^{(1)}+\text{max}\{\tau_c^{(2)}+t_c^{(2)}, \tau_b^{(1)}+t_b^{(1)}\}.
\end{split}
\end{equation}
Since some communication costs being overlapped overlap with the computation, the non-overlapped communication cost, $t_{c}^{no}$, becomes the bottleneck of the system. In WFBP, we redefine $r=\frac{t_c^{no}}{t_f+t_b}$, so the main problem of WFBP is that when the communication cannot be fully overlapped by computation, i.e., $\tau_c^{(l+1)}+t_c^{(l+1)} >\tau_b^{(l)}+t_b^{(l)}$, $t_c^{no}$ will limit the system scalability.
\begin{figure}[!ht]
	\centering
	\includegraphics[width=\linewidth]{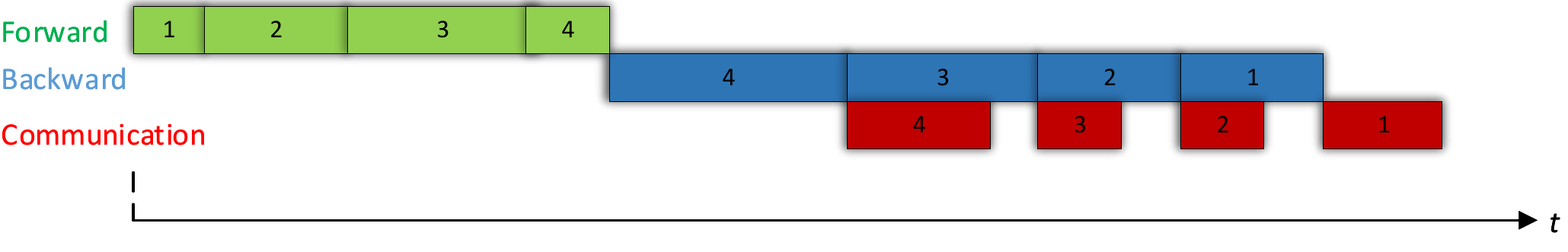}
	\vspace{-20pt}
	\caption{The time-line of WFBP for a 4-layer neural network with communication overheads. Gradient communication of each layer begins immediately after the backward step of that layer.}
	\label{fig:wfbp}
	\vspace{-10pt}
\end{figure}

\subsection{Communication model}
In Eq. \ref{equ:ssgd}, we use $\Delta W_i=\sum_{g=1}^{N}\nabla\mathcal{L}(W_{i},D_{i}^{g})$ to represent the aggregation of gradients from $N$ workers, which is an all-reduce operation. There are many optimized algorithms for the all-reduce operation with different number of processes and message sizes \cite{rabenseifner2004optimization}\cite{thakur2005optimization}\cite{hoefler2010toward}. To simplify the problem, we assume that the number of GPUs is power-of-two, and the peer to peer communication cost is modeled as $\alpha+\beta M$ \cite{sarvotham2001connection}, where $\alpha$ is the latency component, $\beta$ is the communication speed, and $M$ is the message size. Without loss of generality, we do not limit the communication model to one specific algorithm. Given a constant number of GPUs $N$, the time cost of all-reduce can be generalized as
\begin{equation}\label{equ:tcomm}
T_{ar}(M)=a+bM,
\end{equation}
where $a$ and $b$ are two constant numbers that are not related to $M$. 
Some optimized all-reduce algorithms are summarized in Table \ref{table:allreduce}.

\begin{table}[!ht]
		\centering
		\caption{Cost of different all-reduce algorithms}
		\label{table:allreduce}
		\begin{tabular}{|l|c|c|}
			\hline
			All-reduce Algorithm &  $a$ & $b$ \\\hline
			\hline
			Binary tree & $2\alpha log_2N$ & $(2\beta+\gamma)log_2N$ \\\hline
			Recursive doubling& $\alpha log_2N$ & $(\beta+\gamma)log_2N$  \\\hline
			Recursive halving and doubling& $2\alpha log_2N$ & $2\beta-\frac{1}{N}(2\beta+\gamma)+\gamma$ \\\hline
			Ring & $2(N-1)\alpha$ & $\frac{2(N-1)}{N}\beta+\frac{(N-1)}{N}\gamma$  \\\hline
		\end{tabular}
\end{table}
With a given hardware configuration (i.e., $N, \alpha, \beta$, and $\gamma$ are fixed), the time cost of the all-reduce operation is a linear function of the variable $M$. The linear function has an y-intercept $a$ and a scope $b$.

One important property of WFBP is that the messages are communicated layer by layer, which means that it needs to do many all-reduce operations. In each all-reduce, however, there is an extra cost of $a$ which is not related with $M$. Importantly, the linear function with a positive y-intercept value has a property of
\begin{equation}\label{equ:pro}
T_{ar}(M_{1})+T_{ar}(M_{2}) > T_{ar}(M_1+M_2).
\end{equation}
In other words, communicating an $M_1+M_2$ bytes of message is more efficient than communicating an $M_1$ message and an $M_2$ message separately.

\section{Problem Formulation}\label{s:profor}

\subsection{Problem formulation}
As we have shown that merging the gradients can improve the communication efficiency in Eq. \ref{equ:pro}, we further discuss three cases of WFBP to explore under what scenarios can we merge gradients to reduce the iteration time. The three cases are illustrated in Fig. \ref{fig:pipeline}.
\begin{figure}[!ht]
	\centering
	\includegraphics[width=\linewidth]{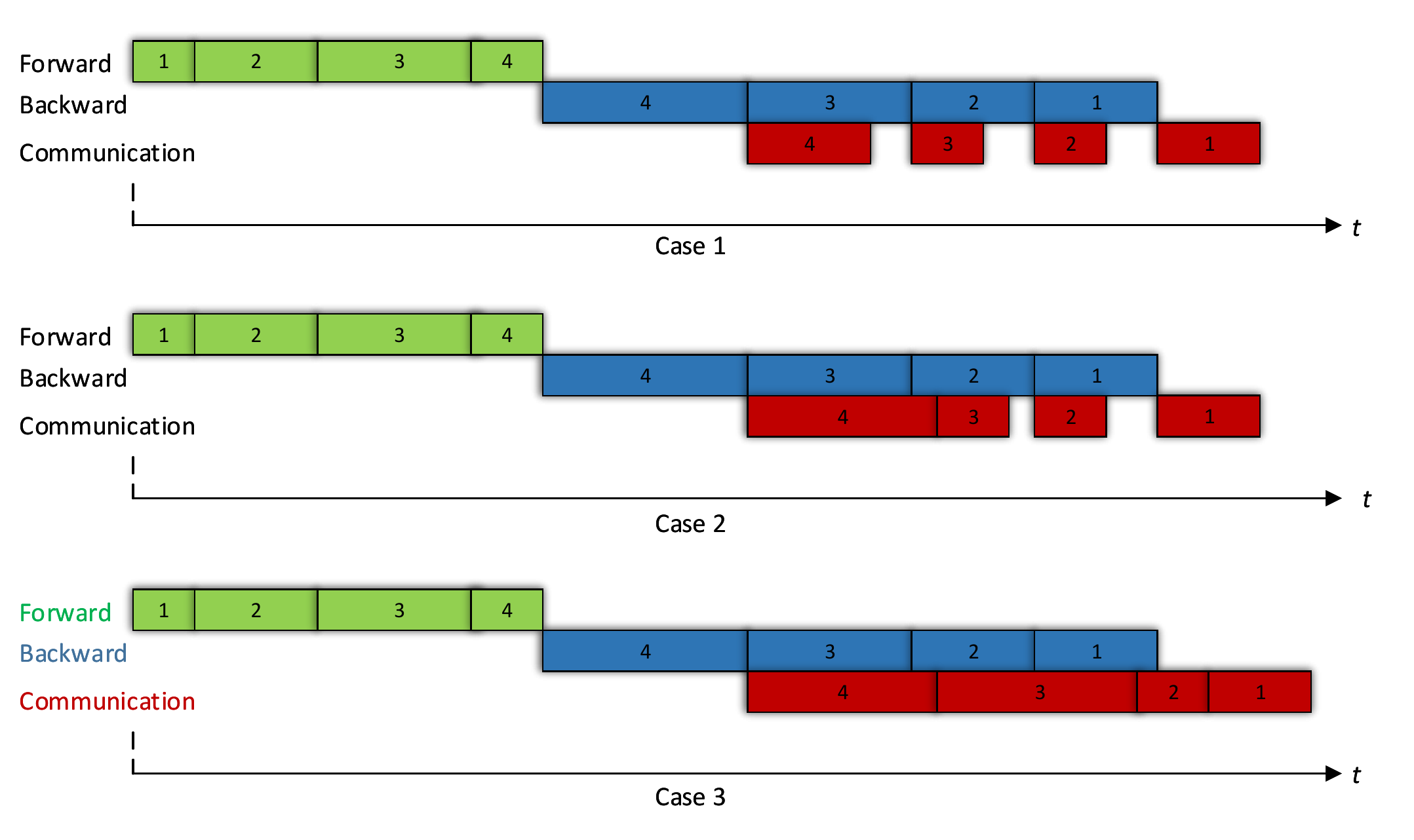}
	\vspace{-20pt}
	\caption{Three cases of WFBP.}
	\label{fig:pipeline}
%	\vspace{-10pt}
\end{figure}

\textbf{Case 1}. In the ideal case, for $2\leq l \leq L, t_{c}^{(l)}\leq t_{b}^{(l-1)}$. The overhead of gradient communication is totally hidden by computation so that it is not necessary to merge the gradients. The iteration time is
\begin{equation}\label{equ:opt}
t_{iter}=t_{f}+t_{b}+t_{c}^{(1)}.
\end{equation}

\textbf{Case 2}. There exists a layer (e.g., the $C^{th}$ layer) whose communication time cannot be totally overlapped by computation, i.e., $t_{c}^{(C)} > t_{b}^{(C-1)}$, but all its following layers' communication can be fully overlapped. In other words, there exists a $K$ ($C<K$), such that $\sum_{j=K}^{C-1}t_{c}^{(j)}<\sum_{j=K}^{C-1}t_{b}^{(j-1)}$, which means from $K^{th}$ to $C^{th}$ layers, the communication can be fully overlapped. As shown in the second sub-figure in Fig. \ref{fig:pipeline}, the communication time of layer $4$ is larger than the computation time of layer $3$, while the total communication time of layer 4 and layer 3 is shorter than the total computation time of layer 3 and layer 2. The communication time can also be totally hidden, which is the same with Case 1, so we have
\begin{equation}
t_{iter}=t_{f}+t_{b}+t_{c}^{(1)}.
\end{equation}
In both Case 1 and Case 2, $r=\frac{t_c^{(1)}}{t_f+t_b}$, which is generally a small value because there is only an extra communication overhead from layer 1.

\textbf{Case 3}. Contrary to Case 2, if $t_{c}^{(C)} > t_{b}^{(C-1)}$, there does not exist a $K$, where $2\leq K<C$, such that $\sum_{j=K}^{C-1}t_{c}^{(j)}<\sum_{j=K}^{C-1}t_{b}^{(j-1)}$. In other words, from the $2^{nd}$ to the $C^{th}$ layer, the sum of communication costs can not be totally overlapped, so that the communication becomes the bottleneck. We have
\begin{equation}
t_{iter}=t_{f}+\sum_{j=C-1}^{L}t_{b}^{(j-1)}+\sum_{j=1}^{C}t_{c}^{(j)}.
\end{equation}
Both Case 1 and Case 2 are ideal cases that the overhead of communication can be easily hidden. In the high latency or low bandwidth network environment, Case 3 could be more often happened. The main problem of Case 3 is that many layers' communication overheads cannot be hidden by the computation.

From the property of Eq. \ref{equ:pro}, two or more small messages can be merged to one larger size message before being exchanged. In other words, there exists an $m$, where $2\leq m < C$, such that we can merge the gradients from the $m^{th}$ layer to  the $C^{th}$ layer, and the merged gradients can be communicated with a cost that can be hidden by computation or smaller than the original cost. An example is shown in Fig. \ref{fig:optpipeline}.
\begin{figure}[!ht]
	\centering
	\includegraphics[width=\linewidth]{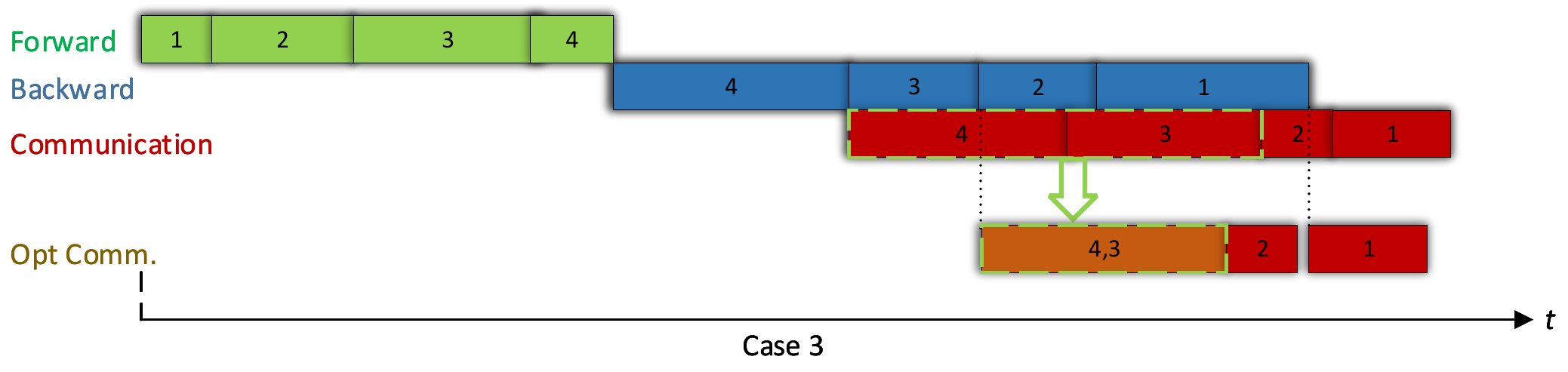}
	\vspace{-20pt}
	\caption{Merged gradient communication.}
	\label{fig:optpipeline}
%	\vspace{-10pt}
\end{figure}
The gradients of layer $4$ and layer $3$ are merged into one message to be exchanged with other nodes, which could reduce the overall transmission time. As a result, the communication of layer $2$ can also be finished before the computation of layer $1$.

In order to reduce the cost of communication, we need to find the optimal $\mathbb{M}$ so that for $m\in \mathbb{M}$, it has $t_{c}^{(C,m)}+\sum_{j=2}^{m-1}t_{c}^{(j)}\leq \sum_{j=1}^{m-1}t_{b}^{(j)}$. There may exist two scenarios: (1) If such $m$ exists, then the iteration time $t_{iter}$ becomes the same as Eq. \ref{equ:opt}. (2) If such $m$ does not exist, then we need to minimize the total cost of communication. In summary, we need to find the set of layers $\mathbb{M}=\{l|2\leq l \leq C\}$, whose gradients are merged with their previous layers and be communicated together, such that the iteration time is minimal.

\begin{definition}{(Merged-gradient).}
A layer $l$ is called a merged-gradient layer if at the timestamp of $\tau_c^{(l)}$, instead of communicating the gradients of that layer, merging its gradients to its previous layer $l-1$ to be communicated together. The operator $\oplus$ defines the gradients merging between two consecutive layers, say $(l)\oplus (l-1)$.
\end{definition}
\begin{condition}
If layer $l$ is a merged-gradient layer, then the following three conditions hold.
\begin{equation}\label{ass:1}
t_{c}^{(l)}=0
\end{equation}
\begin{equation}\label{ass:2}
\tau_c^{(l-1)}=\text{max}\{\tau_c^{(l)}, \tau_b^{(l-1)}+t_b^{(l-1)}\}
\end{equation}
\begin{equation}\label{ass:3}
p^{(l-1)}=p^{(l-1)}+p^{(l)}
\end{equation}
\end{condition}
Condition \ref{ass:1} indicates that layer $l$ does not need to be communicated by itself so it has zero overhead. Condition \ref{ass:2} is obvious by Eq. \ref{equ:startt}. Condition \ref{ass:3} indicates that the gradients of $l$ should be merged to layer $l-1$ such that the number of updated parameters of layer $l-1$ becomes the summation of layer $l$ and layer $l-1$.

From Eq. \ref{equ:comm}, the communication time of each layer is represented by
\begin{equation}\label{equ:comm}
t_{c}^{(l)}=T_{ar}(p^{(l)}).
\end{equation}
To be more generalized, the time cost of backward computation is modeled as a function with respect to the FLOPS (floating-point operation per second) of the processor $G$, the number of parameters and some other factors $\theta$ (e.g., batch size and type of layers) \cite{hang2017paleo}, say
\begin{equation}\label{equ:comp}
t_b^{(l)}=T_b(p^{(l)}, G, \theta)
\end{equation}
Plug Eq. \ref{equ:comm} and Eq. \ref{equ:comp} in Eq. \ref{equ:startt} and Eq. \ref{equ:startcomp}, we obtain
\begin{equation}\label{equ:startcompfinal}
\tau_b^{(l)}=
\begin{cases}
t_f & l=L\\
\tau_b^{(l+1)}+T_b(p^{(l+1)}, G, \theta) & 1\leq l<L
\end{cases}.
\end{equation}
And the start time of communication of each layer can be represented by
\begin{equation}\label{equ:starttfinal}
\tau_c^{(l)}=
\begin{cases}
\tau_b^{(l)}+T_b(p^{(l)}, G, \theta) & l=L\\
\text{max}\{\tau_c^{(l+1)}+T_{ar}(p^{(l+1)}), \tau_b^{(l)}+t_b^{(l)}\} & 1\leq l<L
\end{cases}.
\end{equation}
The iteration time can be rewritten as
\begin{equation}\label{equ:iterfinal}
t_{iter}=T_{ar}(p^{(1)})+\text{max}\{\tau_c^{(2)}+T_{ar}(p^{(2)}), \tau_b^{(1)}+T_b(p^{(1)}, G)\}.
\end{equation}
So we can formulate the problem as follows. Given a DNN with $L$ learnable layers training with S-SGD across $N$ nodes, we want to find a set of merged-gradient layers
\begin{equation}
\mathbb{M}=\{l|\text{layer } l \text{ is a merged-gradient layer, and }2\leq l\leq L\},
\end{equation}
such that the iteration time $t_{iter}$ of Eq. \ref{equ:iterfinal} is minimal. In a specific cluster of $N$ nodes, and each node has $G$ FLOPS of computation capability, we can remove the notation of $G$ in Eq. \ref{equ:iterfinal}. At last, we formulate the following optimization problem:
\begin{equation}\label{equ:prob}
\text{minimize } T_{ar}(p^{(1)})+\text{max}\{\tau_c^{(2)}+T_{ar}(p^{(2)}), \tau_b^{(1)}+T_b(p^{(1)}))\}.
\end{equation}

\section{Solution}\label{s:method}

In this section, we first perform some theoretical analysis on the optimization problem, and then we propose an optimal and efficient solution to the problem.

\subsection{Theoretical analysis}
The first term of Eq. \ref{equ:prob} can be neglected to find the solution, so the objective function becomes
\begin{equation}\label{equ:newopt}
\begin{split}
t&=\text{max}\{\tau_c^{(2)}+T_{ar}(p^{(2)}), \tau_b^{(1)}+T_b(p^{(1)})\}\\
&=\text{max}\{
\text{max}\{\tau_c^{(3)}+T_{ar}(p^{(3)}), \tau_b^{(2)}+T_b(p^{(2)})\}+T_{ar}(p^{(2)}),\\
& \tau_b^{(1)}+T_b(p^{(1)})\}
\end{split}
\end{equation}

Assume that layer $3$ is a merged-gradient layer, we have $t_c^{(3)}=0$ and $t_c^{(2)}=T_{ar}(p^{(2)}+p^{(3)})$. We plug in these two new values to the above equation. Thus,
\begin{equation}\label{equ:merged}
\begin{split}
\hat{t}&=\text{max}\{
\text{max}\{\tau_c^{(3)}, \tau_b^{(2)}+T_b(p^{(2)})\}+T_{ar}(p^{(2)}+p^{(3)}),\\
& \tau_b^{(1)}+T_b(p^{(1)})\}
\end{split}
\end{equation}
Compare Eq. \ref{equ:newopt} to Eq. \ref{equ:merged}, we want to prove that in what conditions $\hat{t}< t$, i.e., layer $3$ can be a gradient-merged layer. In Eq. \ref{equ:newopt} and Eq. \ref{equ:merged}, the second term of outer max function is the same, so we can just compare the first term, i.e.,
\begin{align*}
\text{max}\{\tau_c^{(3)}, \tau_b^{(2)}+T_b(p^{(2)})\}+T_{ar}(p^{(2)}+p^{(3)}) < \\
\text{max}\{\tau_c^{(3)}+T_{ar}(p^{(3)}), \tau_b^{(2)}+T_b(p^{(2)})\}+T_{ar}(p^{(2)}).
\end{align*}
Since $\tau_b^{(1)}=\tau_b^{(2)}+T_b(p^{(2)})$, we simplify the above inequality and obtain
\begin{equation}\label{equ:ineq}
\begin{split}
\text{max}\{\tau_c^{(3)}, \tau_b^{(1)}\}+T_{ar}(p^{(2)}+p^{(3)}) < \\
\text{max}\{\tau_c^{(3)}+T_{ar}(p^{(3)}), \tau_b^{(1)}\}+T_{ar}(p^{(2)}).
\end{split}
\end{equation}
To eliminate the max function, we consider the following three conditions.

\textbf{C.1}. $\tau_b^{(1)}<\tau_c^{(3)}$. The inequality \ref{equ:ineq} becomes
\begin{align*}
\tau_c^{(3)}+T_{ar}(p^{(2)}+p^{(3)}) <\tau_c^{(3)}+T_{ar}(p^{(3)})+T_{ar}(p^{(2)}),
\end{align*}
i.e.,
\begin{align*}
T_{ar}(p^{(2)}+p^{(3)}) <T_{ar}(p^{(3)})+T_{ar}(p^{(2)}),
\end{align*}
which holds due to the property of $T_{ar}$ in Eq. \ref{equ:pro}.

\textbf{C.2}. $\tau_c^{(3)}\leq \tau_b^{(1)}<\tau_c^{(3)}+T_{ar}(p^{(3)})$. Then we obtain
\begin{align*}
\tau_b^{(1)}+T_{ar}(p^{(2)}+p^{(3)}) <\tau_c^{(3)}+T_{ar}(p^{(3)})+T_{ar}(p^{(2)}).
\end{align*}
According to Eq. \ref{equ:tcomm},
\begin{align*}
T_{ar}(p^{(3)})+T_{ar}(p^{(2)})-T_{ar}(p^{(2)}+p^{(3)})=a,
\end{align*}
then the above inequality holds under the condition
\begin{align*}
\tau_b^{(1)}-\tau_c^{(3)} < a.
\end{align*}

\textbf{C.3}. $\tau_b^{(1)}\geq \tau_c^{(3)}+T_{ar}(p^{(3)})$. The inequality \ref{equ:ineq} becomes
\begin{align*}
\tau_b^{(1)}+T_{ar}(p^{(2)}+p^{(3)}) < \tau_b^{(1)}+T_{ar}(p^{(3)})+T_{ar}(p^{(2)}),
\end{align*}
i.e., $
%\begin{align*}
T_{ar}(p^{(2)}+p^{(3)}) < T_{ar}(p^{(3)})+T_{ar}(p^{(2)}),$
%\end{align*}
which is in contradiction with Eq. \ref{equ:tcomm}. So under \textbf{C.3}, layer $3$ should not be a merged-gradient layer.

\begin{theorem}\label{theorem:opt}
Given an $L$-layer DNN which is trained with S-SGD in a cluster of $N$ nodes, if the gradient communication is done through all-reduce, one can find all the merged-gradient layers $\mathbb{M}$ such that the iteration time is minimal, and
\begin{equation}\label{the:solution}
\mathbb{M}=
\{
l|\tau_b^{(l-2)}-\tau_c^{(l)} < a, \text{and }1<l\leq L
\}.
\end{equation}
\end{theorem}

\begin{proof}
For the $l^{th}$ layer, where $1<l\leq L$, since the layer $l$ has only two choices (merged-gradient layer or NOT merged-gradient). If $\tau_b^{(l-2)}-\tau_c^{(l)} < b$, we just need to prove the communication of its previous layer $l-1$ can be finished earlier if layer $l$ is a merged-gradient layer. I.e.,
\begin{equation}
\begin{split}
\text{max}\{\tau_c^{(l)}, \tau_b^{(l-2)}\}+T_{ar}(p^{(l-1)}+p^{(l)}) < \\
\text{max}\{\tau_c^{(l)}+T_{ar}(p^{(l)}), \tau_b^{(l-2)}\}+T_{ar}(p^{(l-1)}).
\end{split}
\end{equation}
As we have discussed in both \textbf{C.1} and \textbf{C.2} conditions, the above inequality holds in the condition of $\tau_b^{(l-2)}-\tau_c^{(l)} < b$. So under this condition, the start time stamp of layer $l-1$ is smaller than that not merged. Consequently, for all $\tau_b^{(l-2)}-\tau_c^{(l)} < a$, we do $(l)\oplus (l-1)$, then Eq. \ref{equ:prob} holds.
\end{proof}

\subsection{Algorithm}
Assume that the $N$-node cluster is connected by an Ethernet with a bandwidth $B$, in which each node has $G$ computation capability. Then the time cost of layer-wise communication can be computed according to Eq. \ref{equ:comm}, and the computation cost of backward propagation can be calculated by Eq. \ref{equ:comp}. Thus, $t_f$, $t_c^{(l)}$ and $t_b^{(l)}$, where $1 \leq l \leq L$, are known. According to Theorem \ref{theorem:opt}, we drive the algorithm to find $\mathbb{M}$ as shown in Algorithm \ref{algo:mgbp}.

\begin{algorithm}[h]
	\caption{Find all merged-gradient layers: $\mathbb{M}$}\label{algo:mgbp}
 	\small
		\textbf{Input: }$a$, $b$, $L$, $N$, $G$, $\bm{p}=[p^{(1)},p^{(2)},...,p^{(L)}]$.\\
		\textbf{Output: $\mathbb{M}$}
	\begin{algorithmic}[1]
		\State Initialize $\bm{tc}[1...L]$; // Communication time cost
		\State Initialize $\bm{tb}[1...L]$; // Backward computation time cost
		\State Initialize $\bm{\tau b}$[1...$L$]; // Backward computation start time
		\State $tf$=$\sum_{l=1}^{L}T_{f}(\bm{p}[l], G)$;
		\For{$l=1\rightarrow L$}
			\State $\bm{tc}[l]=T_{ar}(\bm{p}[l], N)$;
			\State $\bm{tb}[l]=T_{b}(\bm{p}[l], G)$;
		\EndFor
		\State $\bm{\tau b}[L]$=$tf$;
		\For{$l=L-1\rightarrow 2$}
			\State $\bm{\tau b}[l]$ = $\bm{\tau b}[l+1]$ + $\bm{tb}[l+1]$
		\EndFor
		\State $\bm{\tau c}$=\Call{CalculateCommStart}{$\bm{tc}, \bm{tb}, \bm{\tau b}, L$}; 
		\For{$l=L\rightarrow 2$}
		\If{$\bm{\tau b}[l-2]-\bm{\tau c}[l] < a$}
		\State \Call{Merge}{$\bm{\tau b}, \bm{tc}, \bm{p}, l, N$};
		\State $\bm{\tau c}$=\Call{CalculateCommStart}{$\bm{tc}, \bm{tb}, \bm{\tau b}, L$};
		\State $\mathbb{M}$.push($l$);
		\EndIf
		\EndFor
%		\State
%		def merge(self, comms, sizes, i, p, merge_size, comps):
		\Procedure{Merge}{$\bm{\tau b}, \bm{tc}, \bm{p}, l, N$}
		\State $\bm{tc}[l]=0$;
		\State $\bm{p}[l-1]=\bm{p}[l-1]+\bm{p}[l]$;
		\State $\bm{tc}[l-1]=T_{ar}(\bm{p}[l-1], N)$;
		\EndProcedure
%		\State
		\Procedure{CalculateCommStart}{$\bm{tc}, \bm{tb}, \bm{\tau b}, L$}
		\State Initialize $\bm{\tau c}[1...L]$; // Communication start time
		\State $\bm{\tau c}[L]=\bm{\tau b}[L]+\bm{tb}[L]$;
		\For{$l=L-1\rightarrow 1$}
		\State $\bm{\tau c}[l]=\text{max}\{\bm{\tau c}[l+1]+\bm{tc}[l+1], \bm{\tau b}[l]+\bm{tb}[l]\}$;
		\EndFor
		\State \text{Return } $\bm{\tau c}$;
		\EndProcedure
	\end{algorithmic}
\end{algorithm}

The algorithm first (line 1-7) initializes the layer-wise backward computation cost and communication cost according to Eq. \ref{equ:comp} and Eq. \ref{equ:comm} respectively with system settings and bechmarks in the first several iterations. Then (line 8-11) the layer-wise start time of backward computation, which is not changed in its followed computation, is calculated based on Eq. \ref{equ:startcompfinal}, and then iteratively compute the layer-wise start time of communication based on the formula of Eq. \ref{equ:starttfinal}. After that (line 12-16), the merged-gradient layers are found according to Eq. \ref{the:solution}, in which if there is a layer found as a merged-gradient layer, the communication time of its previous layer should be updated according to Eq. \ref{ass:1}, Eq. \ref{ass:2} and Eq. \ref{ass:3}.

The proposed algorithm has a time complexity of $O(L^2)$. For a merged-gradient layer, the algorithm needs to re-calculate the start time of communication of each layer, which is an $O(L)$ search, and it has maximal $L-1$ merged-gradient layers, so the time complexity of the algorithm is $O(L^2)$. Since the algorithm is an one-time calculation at the beginning of the training and it needs not to be re-calculated during the training process, so the overhead of finding $\mathbb{M}$ has no any impact of the training performance.

\begin{algorithm}[h]
	\caption{MG-WFBP S-SGD}\label{algo:gewfbp}
	
	\textbf{Input: } $\bm{D}=[\{X_1, y_1\},...,\{X_n, y_n\}]$, $I$, $net$, $N$, $bs$\\
	\textbf{Output: $\bm{W}=[W^{(1)}, W^{(2)},...W^{(L)}]$}
	\begin{algorithmic}[1]
		\small
		\For{$k=1\rightarrow N$}
		\State Initialize shared and synchronized queue $Q$;
		\State Obtain the parameter size $\bm{p}$ from $net$;
		\State Allocate memories $\bm{W}$;
		\State Initialize $\bm{W}$ in all accelerators;
		\State Get $\mathbb{M}$ from Algorithm \ref{algo:mgbp};
		\State \Call{AsyncHandleComputation}{$Q, \mathbb{M}$};
		\For{$i=1\rightarrow I$}
		\State $di=(i*bs*N)\%n+(k-1)*bs$;
%		\State $\bm{d}=D[di:di+bs]$;
		\State $d=D[di:di+bs]$;
		\State \Call{AsyncHandleComputation}{$Q,d,L$};
		\State WaitForLastCommunicationFinished();
		\State $\bm{W}=\bm{W}-\eta\cdot\nabla \bm{W}$,
		\EndFor
		\EndFor
		\State NotifyFinished(); // Set $isRunning$ to false
		%	\State
		
		\Procedure{AsyncHandleComputation}{$Q,d,L$}
		\State $o=d$;
		\For{$l=1\rightarrow L$}
		\State $o$=FeedForward($l,o$);
		\EndFor
		\For{$l=L\rightarrow 1$}
		\State BackwardPropagation($l$);
		\State $Q.\text{push}(l)$;
		\EndFor
		\EndProcedure
		
		%	\State
		\Procedure{AsyncHandleCommunication}{$Q, \mathbb{M}$}
		\State Initialize $lb$; // layerBuffer
		\While{\textit{isRunning}}
		\State $l=Q.\text{pop()}$;
		\State $lb$.push($l$);
		\If{$l \notin \mathbb{M}$}
		\State SynchonizedAllReduce($lb[0],count$);
		\EndIf
		\If{$l=1$}
		\State NotifyLastCommunicationFinished();
		\EndIf
		\EndWhile
		\EndProcedure
	\end{algorithmic}
	
\end{algorithm}
We denote the WFBP algorithm integrated with the solution $\mathbb{M}$ as MG-WFBP. In MG-WFBP, the merged-gradient layers should be communicated with their previous layers. As a result, MG-WFBP achieves the minimal iteration time of S-SGD under known DNNs and system configurations. The algorithm of MG-WFBP S-SGD is shown in Algorithm \ref{algo:gewfbp}. For each worker, the algorithm first (line 2-6) initializes related variables and calculate $\mathbb{M}$ by using Algorithm \ref{algo:mgbp}. Then (line 7, 24-30) it reads the layer number from the shared queue $Q$ and decides whether its gradients should be communicated. After that (line 9-13), it starts the loop of iteration, and iteratively reads data (line 16-21) to do feed forward operations and backward propagation followed by pushing the layer number into the shared queue. Finally, the algorithm notifies a message of \textit{isRunning=false} to finish training.

%\vspace{-10pt}
\section{Evaluation}\label{s:eval}
We evaluate the performance of MG-WFBP by real experiments on an 8-node GPU cluster with 10GbE, and also by simulations on larger clusters with up to 64 nodes. Two popular CNNs, namely GoogleNet \cite{szegedy2015going} with a batch size of 64 and ResNet-50 \cite{he2016deep} with a batch size of 32, are chosen to test the performance of distributed training with ImageNet dataset ILSVRC-2012 \cite{deng2009imagenet} which includes about 1.28 million training images of 1000 categories. GoogleNet has about 13 millions of parameters, while ResNet-50 has about 25.5 millions. All parameters and gradients are stored as 32-bit single precision floating point numbers.

\subsection{Statistic data}

To verify the communication model in Eq. \ref{equ:comm} empirically, we first present some foregone results including the distribution of layer-wise gradient sizes of evaluated CNNs, and the time of the all-reduce operation all-reduce in a 10Gbps Ethernet (10GbE) using OpenMPI v3.1.1. The distribution of gradient size is shown in Fig. \ref{fig:commoverhead}(a), which shows that the number of parameters of each layer is mainly located in the range of $[10^2,5\times 10^6]$. The measured time of all-reduce under the 10GbE interconnect cluster are shown in Fig. \ref{fig:commoverhead}(b). Take the size of parameters ($4p$ in floating points) as the variable, we can see that the startup overheads (i.e., $2(N-1)\times \alpha$ in the ring-based all-reduce algorithm) are $90.52\mu s$, $271.56\mu s$ and $633.64\mu s$ in 2-, 4- and 8-node clusters with the 10GbE interconnect respectively.

\begin{figure}[!h]
	\centering
		\subfigure[]%\label{fig:gradis}
		{
			\includegraphics[width=0.48\linewidth]{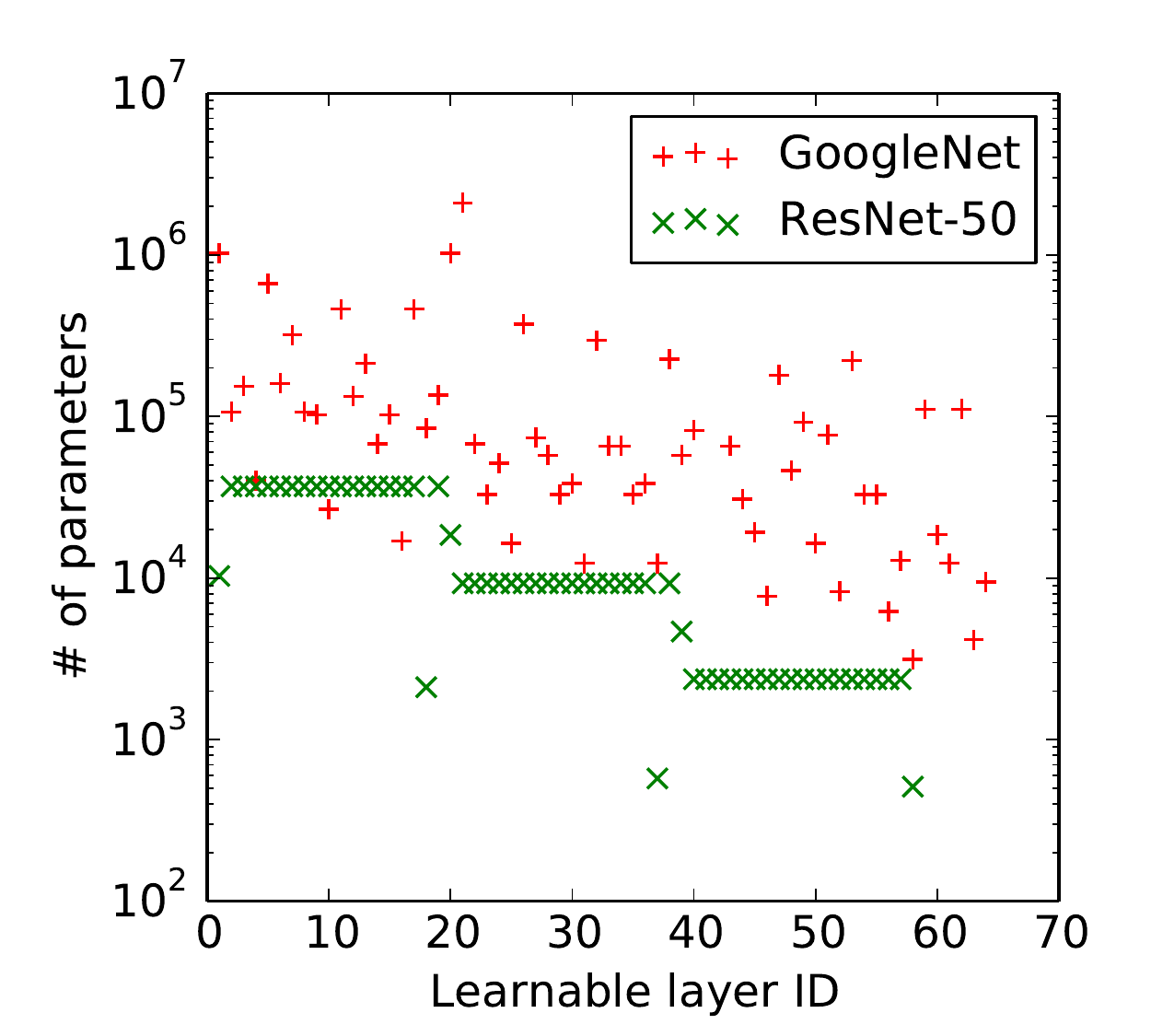}
		}\hspace{-5mm}
		\subfigure[]%\label{fig:commoverhead}
		{
			\includegraphics[width=0.48\linewidth]{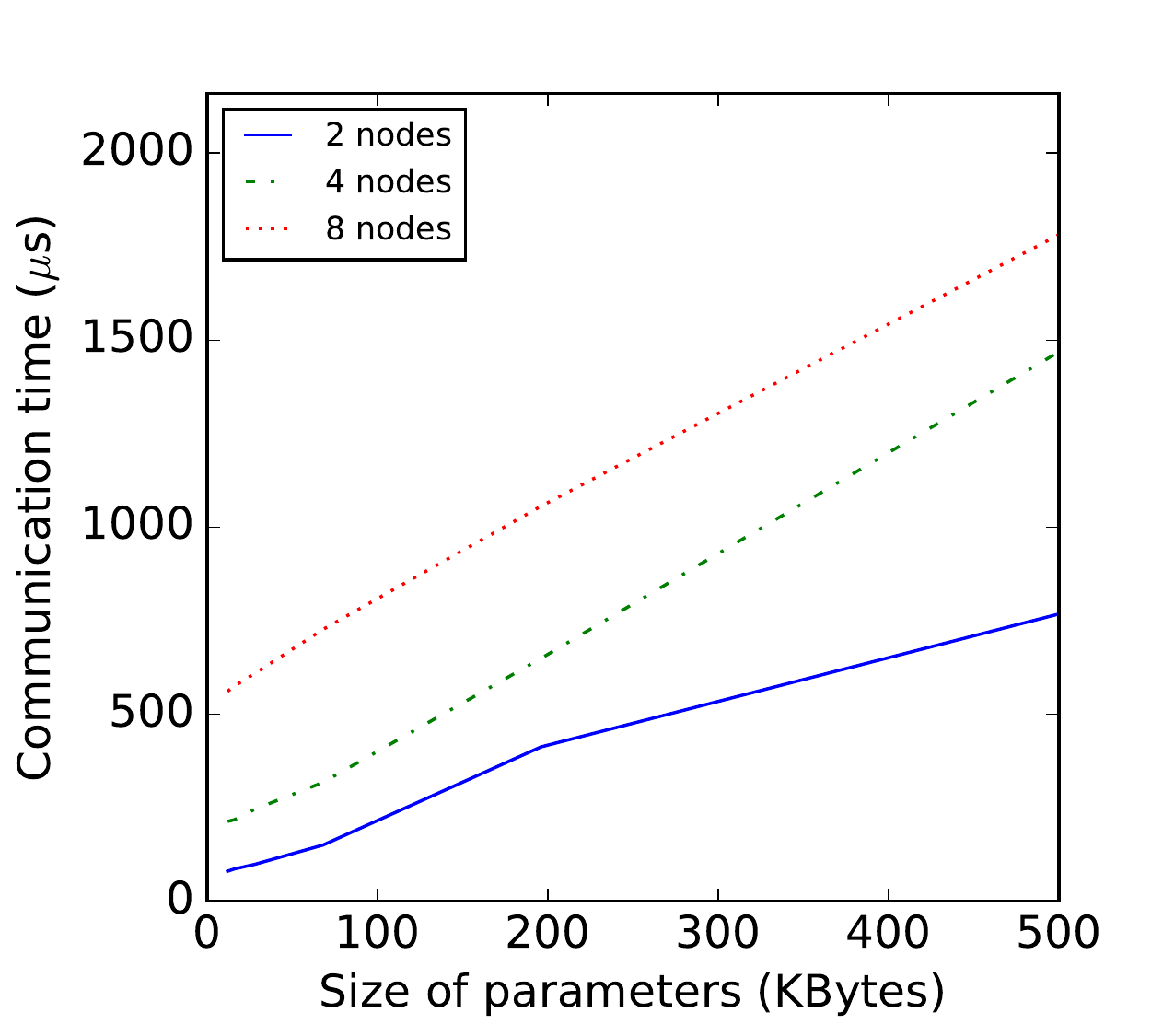}
		}\hspace{-5mm}
\vspace{-5pt}
	\caption{(a) The distribution of layer-wise gradient size of two neural networks. (b) The communication time of the all-reduce along with the size of parameters with the prediction and the measurement.}
	\label{fig:commoverhead}
	\vspace{-10pt}
\end{figure}

\subsection{Real-world experiments}

We integrate WFBP \cite{awan2017s}\cite{zhang2017poseidon}, single-layer communication Sync EASGD (SyncEASGD) \cite{you2017scaling} and our proposed MG-WFBP into B-Caffe\footnote{B-Caffe is an optimized distributed deep learning framework based on Caffe \cite{jia2014caffe}.}, and test the performance across an 8-node GPU cluster with 10GbE. We also compare the scaling efficiencies with TensorFlow. Each server has one Nvidia Tesla K80 card (i.e., 2 GPUs). The OS system is CentOS 7.2, and the major software libraries include CUDA-8.0, cuDNNv6 and NCCL. The compared TensorFlow is at v1.3, and it uses parameter servers to do S-SGD using the official benchmark script\footnote{https://github.com/tensorflow/benchmarks}.

In the real-world experiments, we run 13 epochs to verify the convergence of the CNN training, in which 50000 images are used to test the top-1 accuracies.

\begin{figure}[!h]
	\centering
	\subfigure[Speedup]
	{
		\includegraphics[width=0.46\linewidth]{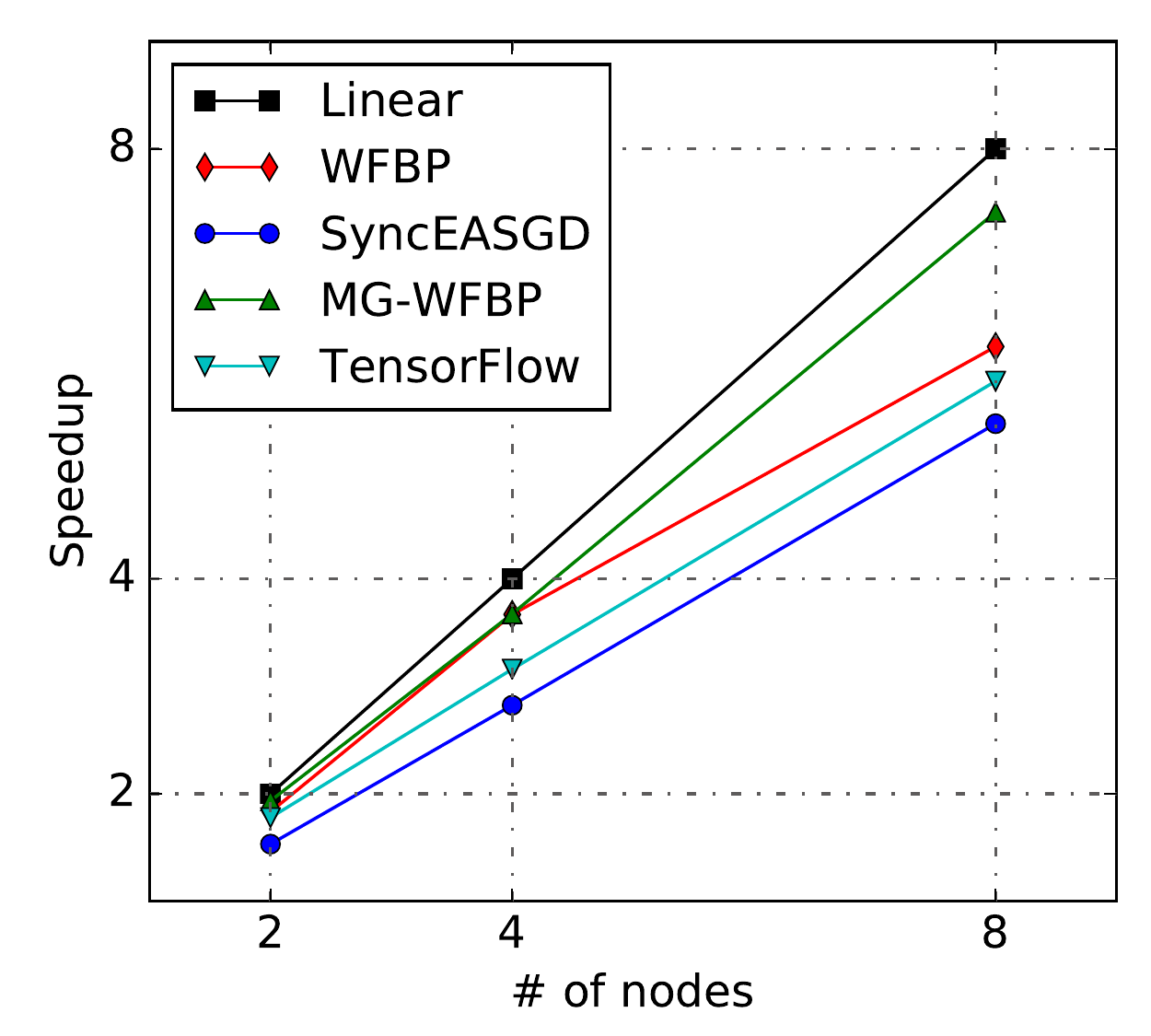}
	}
	\subfigure[Top-1 validation accuracy]
	{
		\includegraphics[width=0.46\linewidth]{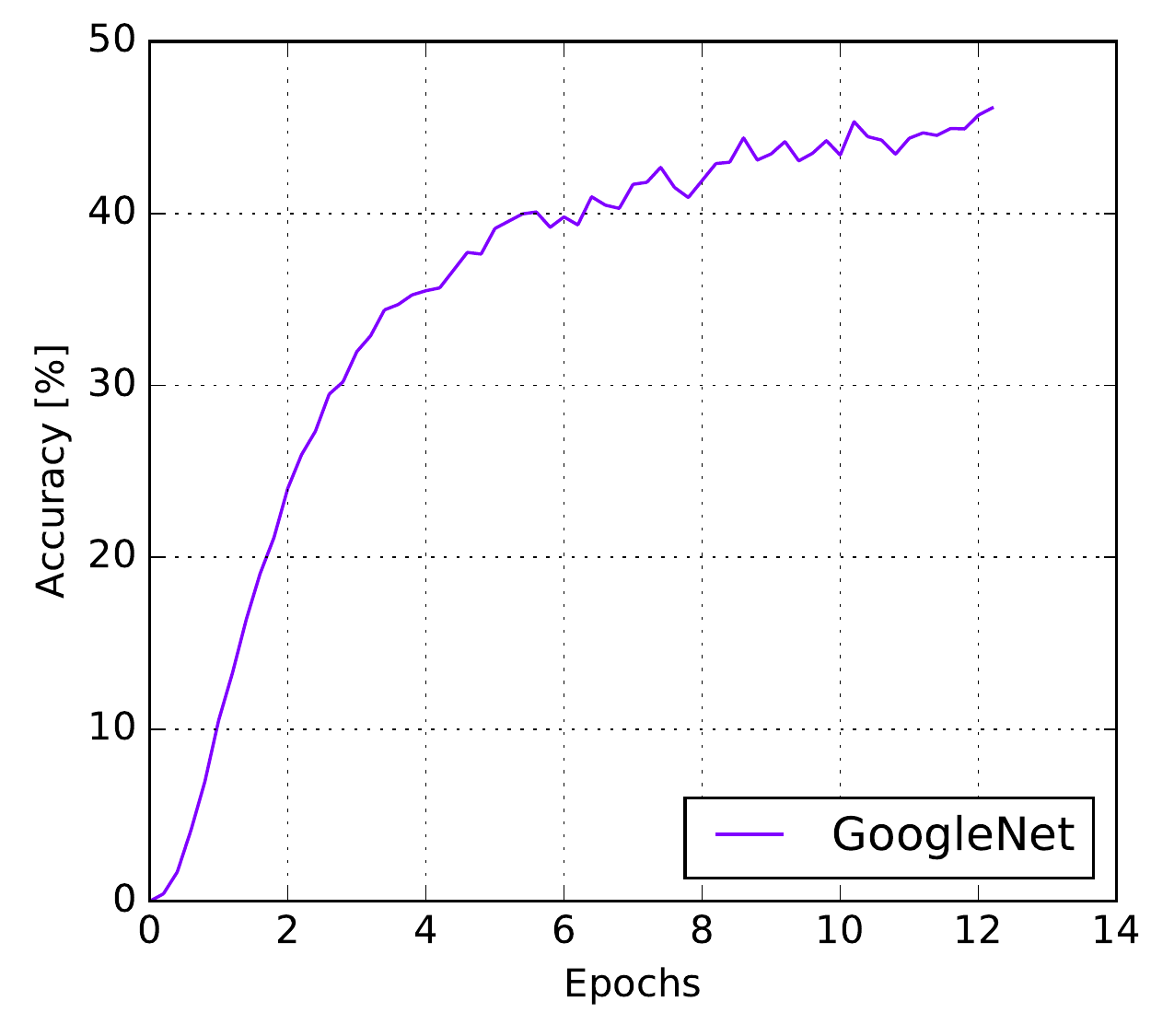}
	}
	\vspace{-5pt}
	\caption{The performance of GoogleNet on the K80 cluster connected with 10GbE. Baseline of the speedup of SGD is on a single machine with 2 GPUs.}
	\label{fig:realresultsgooglenet}
\vspace{-10pt}
\end{figure}

\begin{figure}[!h]
	\centering
	\subfigure[Speedup]
	{
		\includegraphics[width=0.46\linewidth]{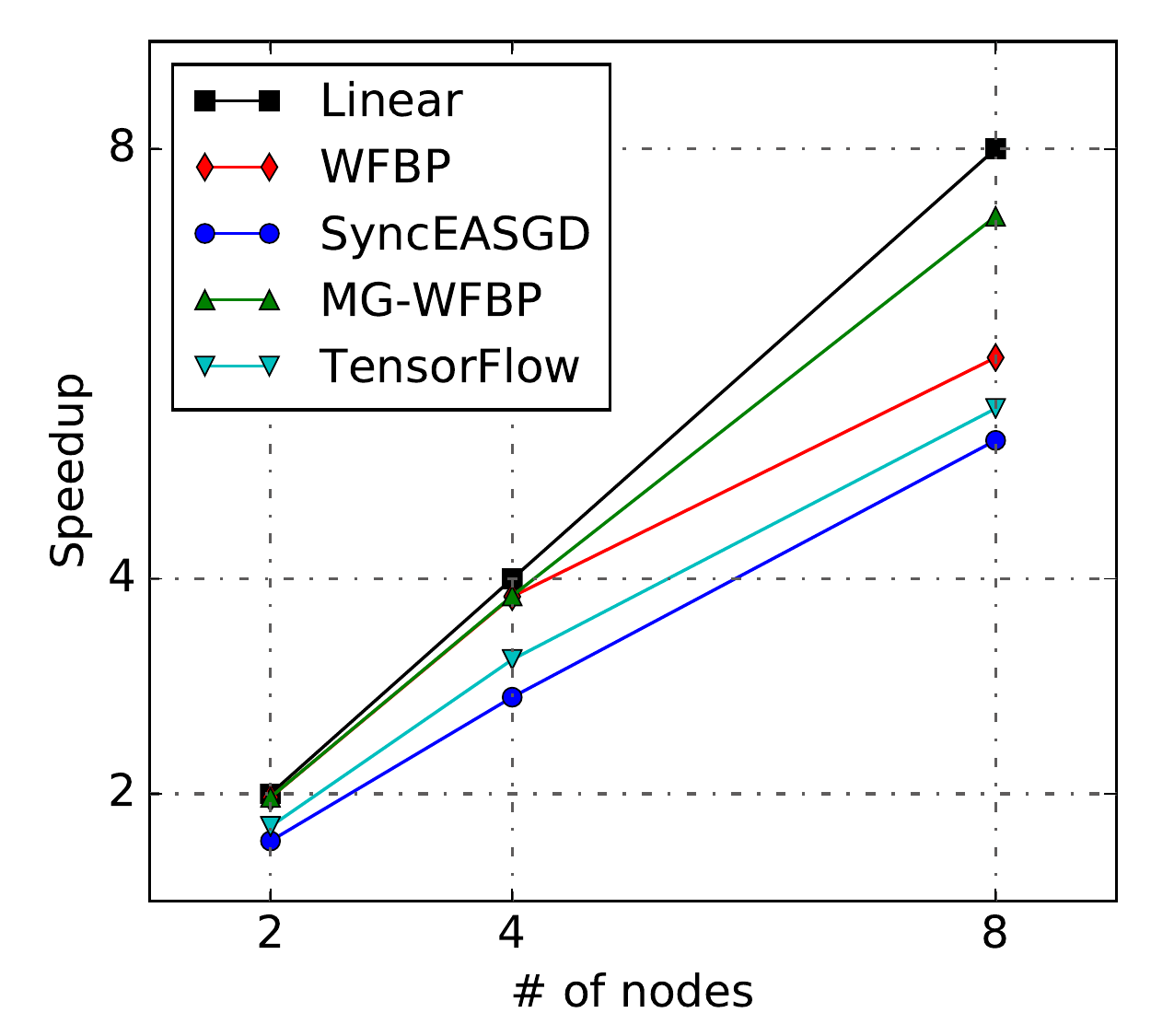}
	}
	\subfigure[Top-1 validation accuracy]
	{
		\includegraphics[width=0.46\linewidth]{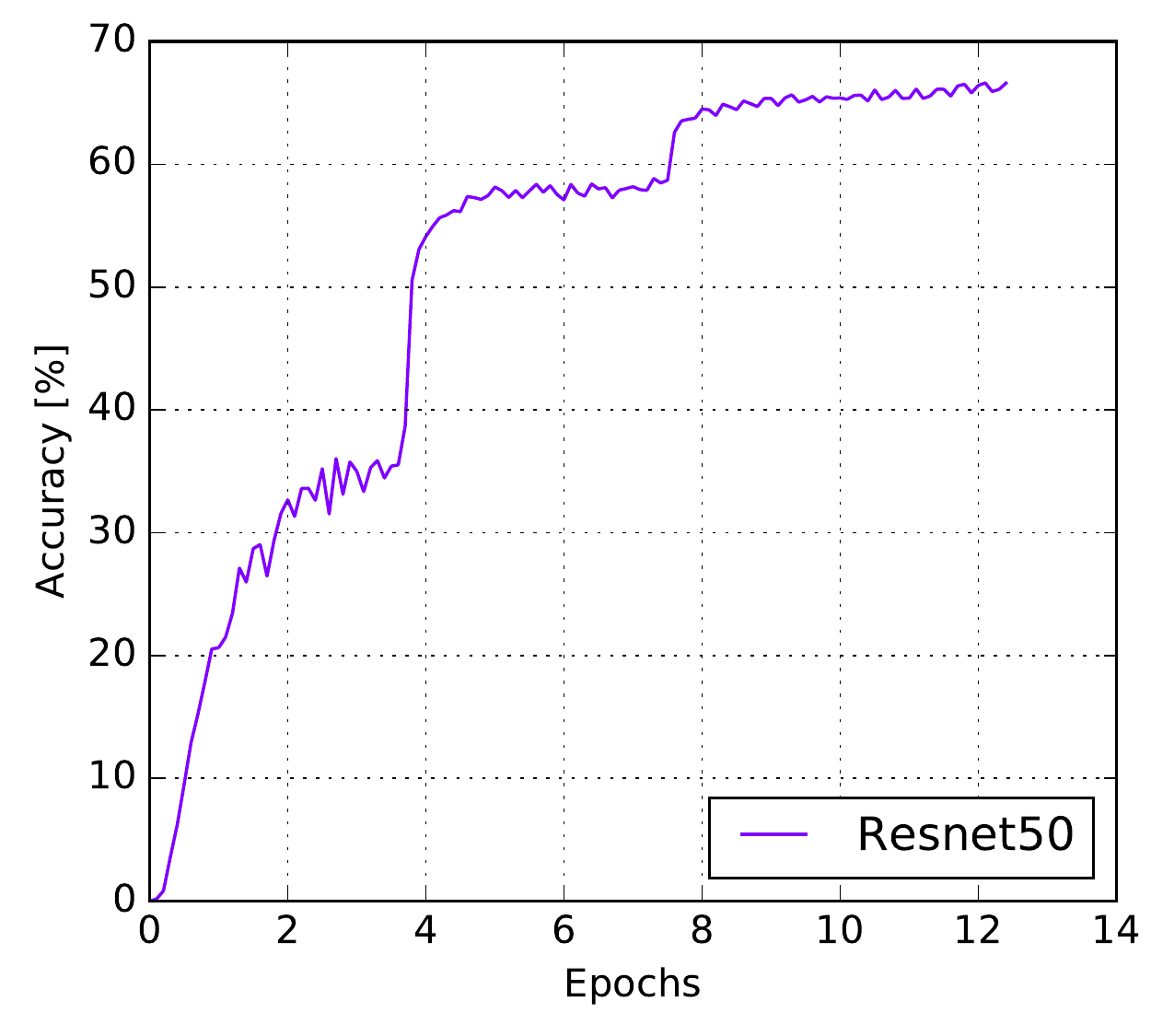}
	}
	\vspace{-5pt}
	\caption{The performance of ResNet-50 on the K80 cluster connected with 10GbE. Baseline of the speedup of SGD is on a single machine with 2 GPUs.}
	\label{fig:realresultsresnet}
	\vspace{-10pt}
\end{figure}

\begin{figure}[!h]
	\centering
	\subfigure[GoogleNet]
	{
		\includegraphics[width=0.48\linewidth]{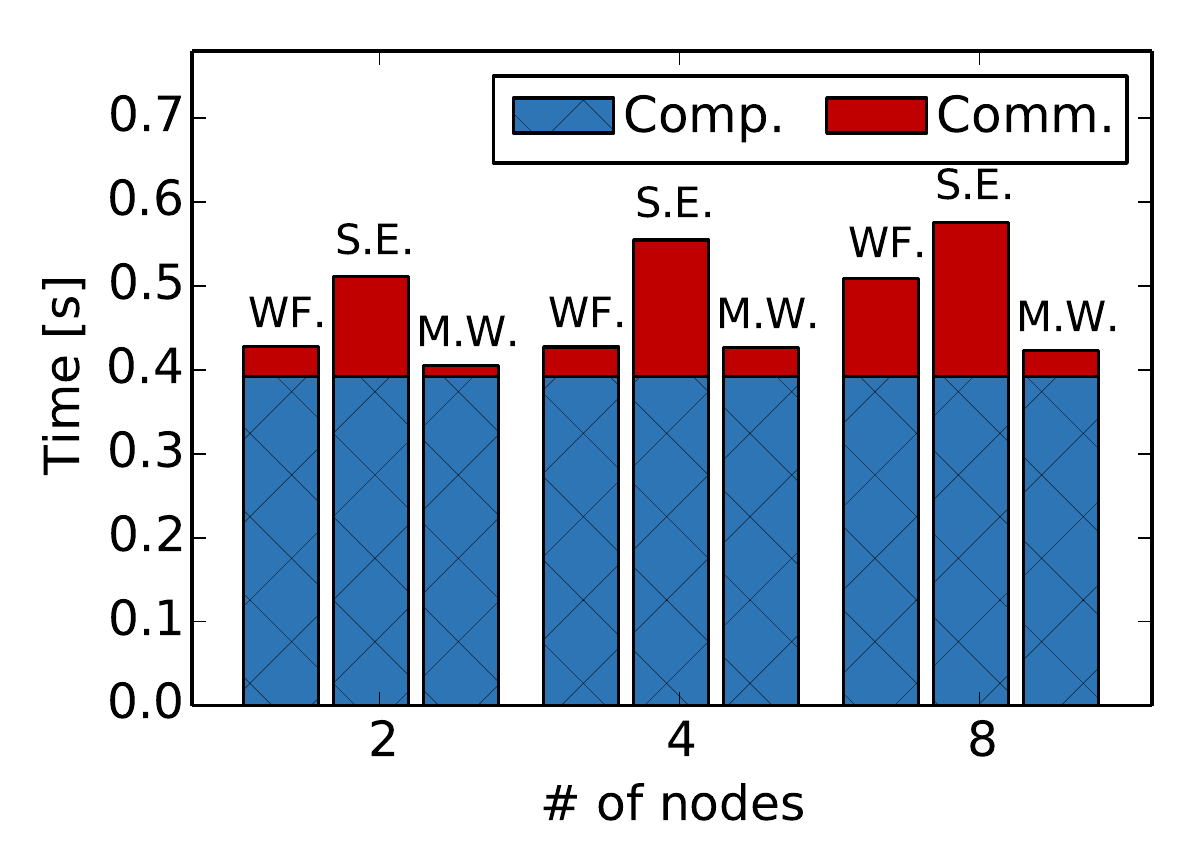}
	}\hspace{-5mm}
	\subfigure[ResNet-50]
	{
		\includegraphics[width=0.48\linewidth]{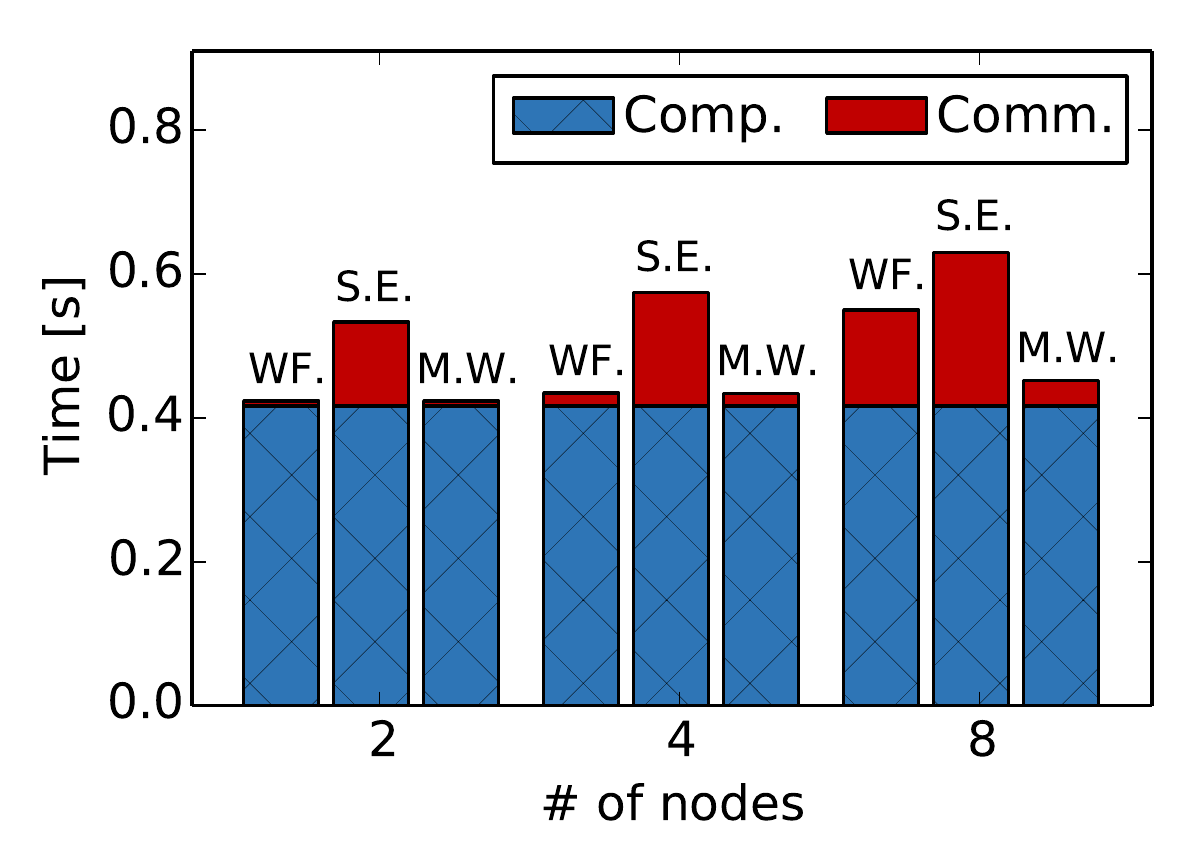}
	}
	\caption{Time costs of non-overlapped communication and computation. `WF.', `S.E.' and `M.W.' indicate WFBP, SyncEASGD and MG-WFBP algorithms respectively. `Comp.' refers to the computation cost (i.e., $t_f+t_b$), and `Comm.' refers to the non-overlapped communication cost (i.e., $t_c^{no}$).}
	\label{fig:realcomm}
%	\vspace{-10pt}
\end{figure}

The experimental results of GoogleNet and ResNet-50 in the K80 cluster are shown in Fig. \ref{fig:realresultsgooglenet} and Fig. \ref{fig:realresultsresnet} respectively. The non-overlapped communication cost compared to the computation time is shown in Fig. \ref{fig:realcomm}. The baseline is the iteration throughput of two GPUs in a single machine, in which no communication via Ethernet is required. And the speedup of throughput on multiple nodes are compared to the baseline. From Fig. \ref{fig:realcomm}, we can observe that for both GoogleNet and ResNet, MG-WFBP performs better than WFBP, SyncEASGD and TensorFlow. SyncEASGD dose not overlap the communication with computation; and hence the communication cost increases when the number of nodes increases. As a consequence, the scaling efficiency of SyncEASGD is poor. WFBP achieves near linear scaling on 2 and 4 nodes, in which the non-overlapped communication overhead are small. When scaling to 8 nodes, however, WFBP has an obvious drop in efficiency due to the increased startup time of layer-wise communication which cannot be totally hidden by computation. Regarding the performance of TensorFlow, it uses parameter servers to do the model aggregation. On one hand, the centralized parameter server based algorithm could easily suffer a bandwidth pressure in the parameter server on the lower speed network \cite{zhang2017poseidon}. On the other hand, it takes two communication directions (workers to PS, and PS to workers) to finish the model synchronization, which introduces more overhead in the synchronization pass. Therefore, though TensorFlow exploits the WFBP technique, the PS-based method performs worse than the decentralized method. Our proposed algorithm has a very small non-overlapped communication cost even on the 8-node cluster, so the scaling efficiency is still close to linear.  In summary, MG-WFBP achieves about $1.2$x and $1.36$x speedups compared to WFBP and SyncEASGD respectively on the 8-node K80 cluster on both GoogleNet and ResNet-50.

\subsection{Simulation}
Due to the hardware limitation, we do not have a very large GPU cluster to support more large-scale experiments. So we conduct simulations based on the real single-GPU performance and the network performance model. Based on the measured layer-wise backward propagation time on the real K80 GPU, we simulate WFBP, SyncEASGD and MG-WFBP by scaling from 4 nodes to 64 nodes.

\begin{figure}[!h]
	\centering
	\subfigure[GoogleNet]
	{
		\includegraphics[width=0.48\linewidth]{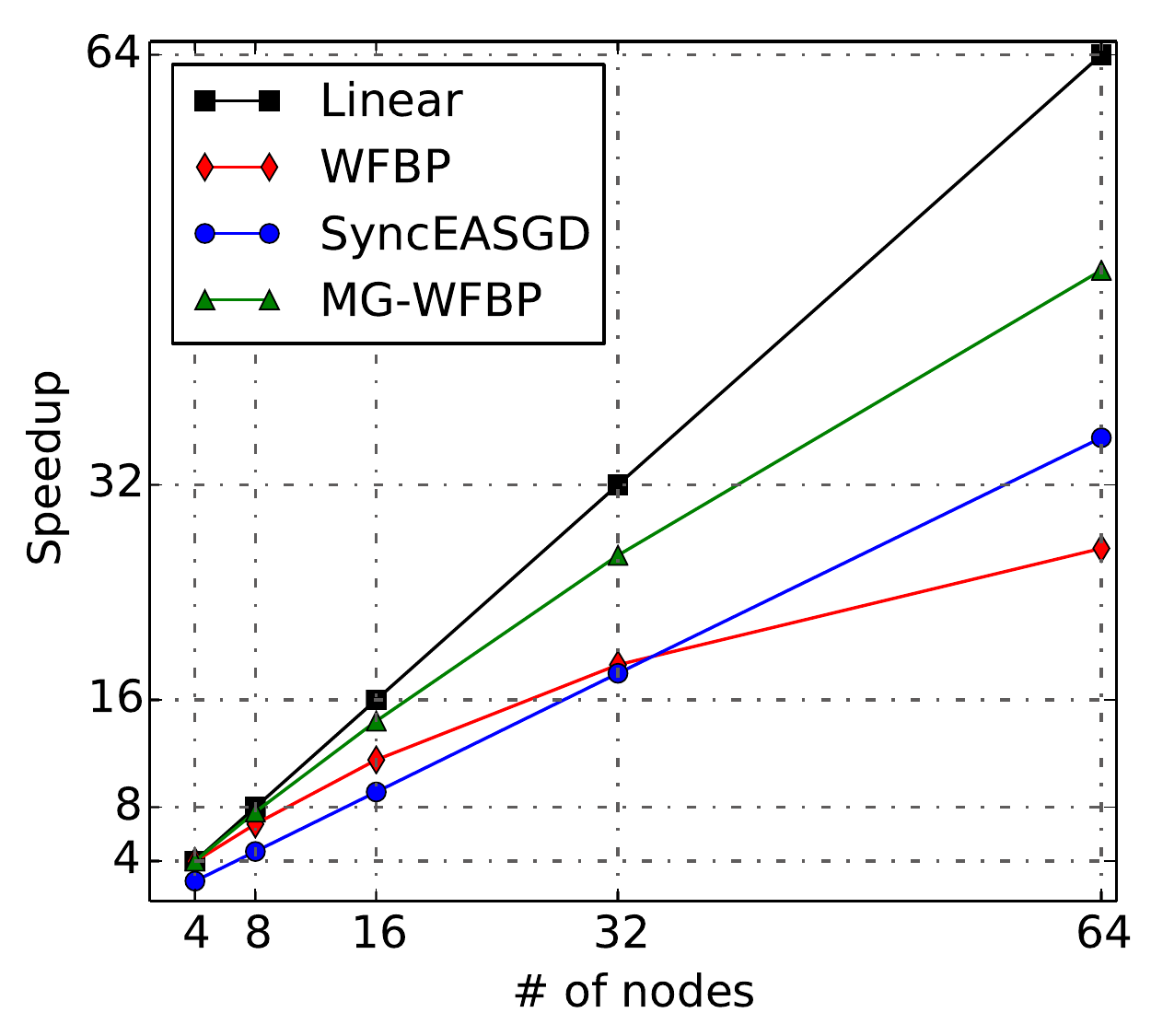}
	}\hspace{-3mm}
	\subfigure[ResNet-50]
	{
		\includegraphics[width=0.48\linewidth]{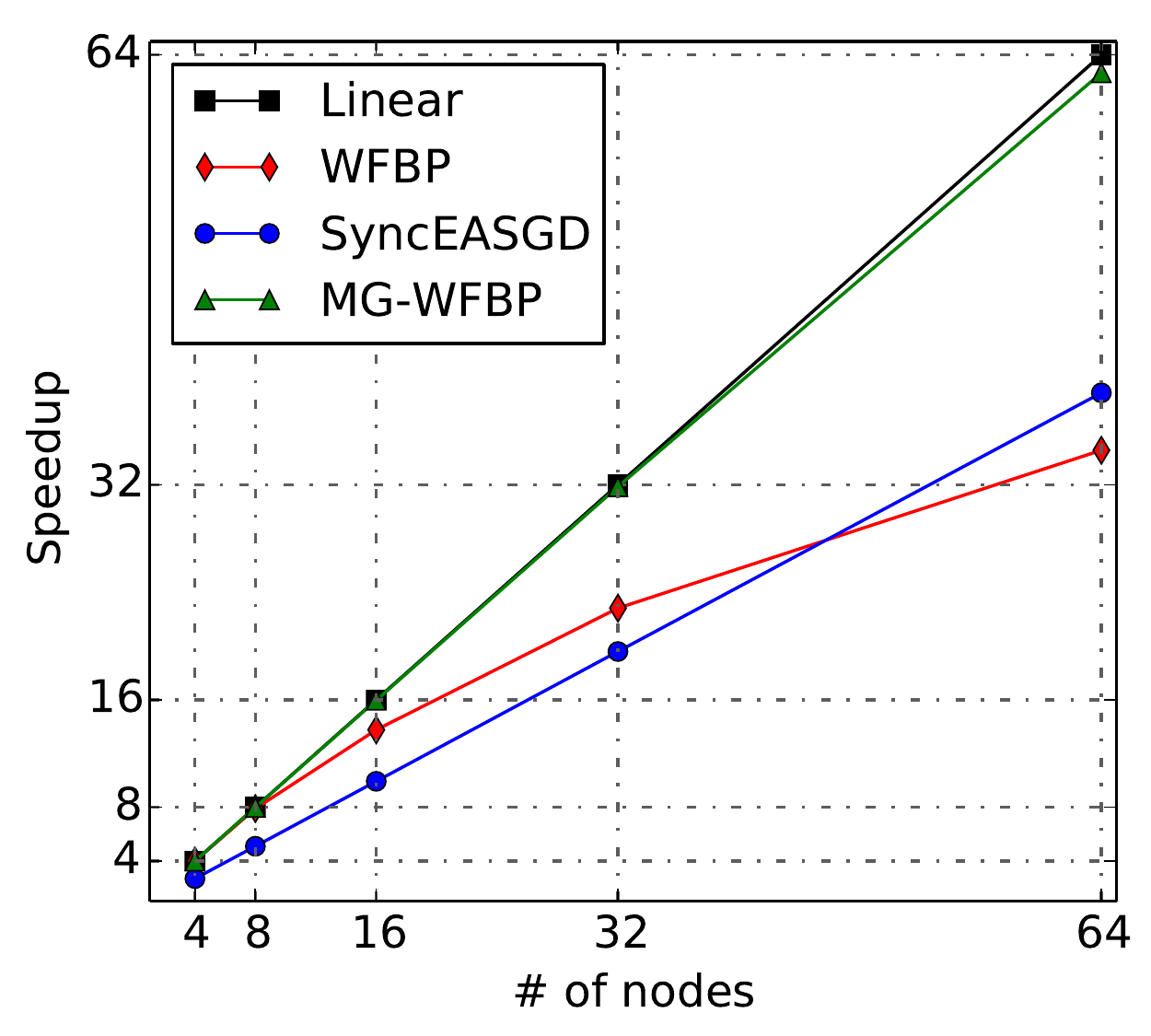}
	}\\
	\vspace{-5pt}
	\caption{The performance comparison on the simulated K80 cluster connected with 10GbE. Baseline of the speedup of SGD is on a single K80.}
	\label{fig:simspeedupk80}
	\vspace{-10pt}
\end{figure}

\textbf{Overall Performance}. We simulate to train GoogleNet and ResNet-50 by scaling from 4 nodes to 64 nodes. The scaling performances are shown in Fig. \ref{fig:simspeedupk80}. On the cluster with K80 GPUs, our proposed algorithm MG-WFBP achieves the best speedup. On the 64-node cluster, MG-WFBP outperforms WFBP and SyncEASGD by $1.78$x and $1.35$x, respectively on GoogleNet. On ResNet-50, MG-WFBP performs almost linear speedup, while WFBP and SyncEASGD only have around $55\%$ scaling efficiency in the 64-node cluster. It is important to notice that the lines of WFBP and SyncEASGD have a crossing point in Fig. \ref{fig:simspeedupk80}. This is because the two algorithms are sub-optimal in utilizing the network bandwidth; when the computation has the opportunity to overlap with communication, and the startup time of network communication is not that large (e.g., 4-16 nodes in the K80 cluster), then WFBP would have the advantage to hide the communication compared to SyncEASGD. But when scaling to large number of nodes (e.g., 64 nodes), the startup time of communication becomes much larger so that it is hard to be hidden, then using a single-layer communication could become a better approach. As we can see, SyncEASGD achieves better scaling efficiency than WFBP in the 64-node cluster on both tested CNNs. MG-WFBP not only overlaps the communication with computation, but it also finds the optimal communication message size. So it achieves better scaling efficiency than SyncEASGD and WFBP. Finally, on training ResNet-50, MG-WFBP achieves about $1.75$x and $1.45$x speedups compared to WFBP and SyncEASGD respectively on the simulated 64-node K80 cluster.

\section{Related Work}\label{s:bm}
The wait-free backward propagation (WFBP) algorithm has recently been proposed to reduce such impact by overlapping communication with computation \cite{awan2017s}\cite{zhang2017poseidon}. In WFBP, the backward computation operations can be started without waiting for the completion of the previous round of data communication. If the communication cost of layer $l+1$ is smaller than the cost of gradients computation of layer $l$, then the communication cost can be completely hidden (except the first layer); and as a result, the scaling efficiency can be close to linear \cite{awan2017s}\cite{zhang2017poseidon}. In practice, however, many DNN models are trained on high-throughput GPUs that result in very short computing time for each backward layer, while it needs to wait for gradient aggregation before starting the next iteration especially on low bandwidth networks (e.g., 10GbE). Current distributed training systems \cite{hoefler2010toward}\cite{jia2018highly} exploit tensor fusion that merges small size of gradients before communicating across workers to reduce the communication overhead. The parameter server (PS) method \cite{li2014communication} is proposed for parallelism between computation and communication, but it easily suffers from the communication traffic jam since PS needs to collect the gradients from all the workers. Sufficient factor broadcasting (SFB) \cite{zhang2017poseidon} uses the matrix factorization technique to reduce the volume of the data that needs to be communicated. Although SFB uses P2P communication to eliminate the bandwidth pressure on the PS, it brings a growing number of sufficient factors with both the increasing number of data samples and workers. Zhang et al. \cite{zhang2017poseidon} proposed Poseidon system with hybrid communication of PS and SFB combined with the WFBP algorithm, and they have achieved 15.5x speed-up on 16 single-GPU (TITANX Pascal) machines. Unfortunately, due to drawbacks of PS and SFB and the communication scheme, Poseidon could also be far away from linear scaling with the number of workers increased due to the communication bottleneck.

In the HPC community, the MPI data communication collectives have been redesigned for distributed training to improve the communication performance across multiple machines \cite{awan2017s}. Many MPI-like implementations, such as OpenMPI\footnote{https://www.open-mpi.org/}, NCCL2\footnote{https://developer.nvidia.com/nccl}, Gloo\footnote{https://github.com/facebookincubator/gloo} and MVAPICH2-GDR\footnote{https://mvapich.cse.ohio-state.edu/}, support efficient CUDA-aware communication between GPUs via network, and many state-of-the-art deep learning frameworks (e.g., TensorFlow, Caffe2 and CNTK) integrate NCCL2 or Gloo for their distributed training modules. Even though these libraries provide very efficient communication collectives, the data communication would still become bottleneck when the communication-to-computation ratio is high, and S-SGD does not scale very well. 

\section{Conclusion}\label{s:conclusion}
In this work, we first show that existing state-of-the-art communication strategies, say wait-free backward propagation (WFBP) and single-layer communication (SyncEASGD), are sub-optimal in the distributed SGD training of deep learning when the communication-to-computation ratio is high. Then we generalize the communication problem as an optimization problem and develop an efficient optimal solution. We then propose the MG-WFBP strategy and implement it in our open-source platform B-Caffe. MG-WFBP achieves a better scalability than WFBP and SyncEASGD in our tested experiments with two popular CNNs (GoogleNet and ResNet-50) across real-world and simulated 10GbE GPU clusters.

\bibliographystyle{IEEEtran}
\Urlmuskip=0mu plus 1mu
\bibliography{merged_gradients.bbl}

\begin{thebibliography}{10}
\providecommand{\url}[1]{#1}
\csname url@samestyle\endcsname
\providecommand{\newblock}{\relax}
\providecommand{\bibinfo}[2]{#2}
\providecommand{\BIBentrySTDinterwordspacing}{\spaceskip=0pt\relax}
\providecommand{\BIBentryALTinterwordstretchfactor}{4}
\providecommand{\BIBentryALTinterwordspacing}{\spaceskip=\fontdimen2\font plus
\BIBentryALTinterwordstretchfactor\fontdimen3\font minus
  \fontdimen4\font\relax}
\providecommand{\BIBforeignlanguage}[2]{{%
\expandafter\ifx\csname l@#1\endcsname\relax
\typeout{** WARNING: IEEEtran.bst: No hyphenation pattern has been}%
\typeout{** loaded for the language `#1'. Using the pattern for}%
\typeout{** the default language instead.}%
\else
\language=\csname l@#1\endcsname
\fi
#2}}
\providecommand{\BIBdecl}{\relax}
\BIBdecl

\bibitem{dean2012large}
J.~Dean, G.~Corrado, R.~Monga, K.~Chen, M.~Devin, M.~Mao, A.~Senior, P.~Tucker,
  K.~Yang, Q.~V. Le \emph{et~al.}, ``Large scale distributed deep networks,''
  in \emph{Advances in neural information processing systems}, 2012, pp.
  1223--1231.

\bibitem{goyal2017accurate}
P.~Goyal, P.~Doll{\'a}r, R.~Girshick, P.~Noordhuis, L.~Wesolowski, A.~Kyrola,
  A.~Tulloch, Y.~Jia, and K.~He, ``Accurate, large minibatch {SGD}: training
  {ImageNet} in 1 hour,'' \emph{arXiv preprint arXiv:1706.02677}, 2017.

\bibitem{watcharapichat2016ako}
P.~Watcharapichat, V.~L. Morales, R.~C. Fernandez, and P.~Pietzuch, ``Ako:
  Decentralised deep learning with partial gradient exchange,'' in
  \emph{Proceedings of the Seventh ACM Symposium on Cloud Computing}.\hskip 1em
  plus 0.5em minus 0.4em\relax ACM, 2016, pp. 84--97.

\bibitem{cui2016geeps}
H.~Cui, H.~Zhang, G.~R. Ganger, P.~B. Gibbons, and E.~P. Xing, ``Geeps:
  Scalable deep learning on distributed {GPUs} with a {GPU}-specialized
  parameter server,'' in \emph{Proceedings of the Eleventh European Conference
  on Computer Systems}.\hskip 1em plus 0.5em minus 0.4em\relax ACM, 2016, p.~4.

\bibitem{alistarh2017qsgd}
D.~Alistarh, D.~Grubic, J.~Li, R.~Tomioka, and M.~Vojnovic, ``{QSGD}:
  Communication-efficient {SGD} via gradient quantization and encoding,'' in
  \emph{Advances in Neural Information Processing Systems}, 2017, pp.
  1707--1718.

\bibitem{lin2017deep}
Y.~Lin, S.~Han, H.~Mao, Y.~Wang, and W.~J. Dally, ``Deep gradient compression:
  Reducing the communication bandwidth for distributed training,'' \emph{arXiv
  preprint arXiv:1712.01887}, 2018.

\bibitem{wen2017terngrad}
W.~Wen, C.~Xu, F.~Yan, C.~Wu, Y.~Wang, Y.~Chen, and H.~Li, ``Terngrad: Ternary
  gradients to reduce communication in distributed deep learning,'' in
  \emph{Advances in Neural Information Processing Systems}, 2017, pp.
  1508--1518.

\bibitem{potluri2013efficient}
S.~Potluri, K.~Hamidouche, A.~Venkatesh, D.~Bureddy, and D.~K. Panda,
  ``Efficient inter-node {MPI} communication using {GPUDirect} {RDMA} for
  {InfiniBand} clusters with {Nvidia} {GPUs},'' in \emph{Parallel Processing
  (ICPP), 2013 42nd International Conference on}.\hskip 1em plus 0.5em minus
  0.4em\relax IEEE, 2013, pp. 80--89.

\bibitem{bayatpour2017scalable}
M.~Bayatpour, S.~Chakraborty, H.~Subramoni, X.~Lu, and D.~K. Panda, ``Scalable
  reduction collectives with data partitioning-based multi-leader design,'' in
  \emph{Proceedings of the International Conference for High Performance
  Computing, Networking, Storage and Analysis}.\hskip 1em plus 0.5em minus
  0.4em\relax ACM, 2017, p.~64.

\bibitem{awan2017s}
A.~A. Awan, K.~Hamidouche, J.~M. Hashmi, and D.~K. Panda, ``S-caffe:
  Co-designing {MPI} runtimes and {Caffe} for scalable deep learning on modern
  {GPU} clusters,'' in \emph{Proceedings of the 22nd ACM SIGPLAN Symposium on
  Principles and Practice of Parallel Programming}.\hskip 1em plus 0.5em minus
  0.4em\relax ACM, 2017, pp. 193--205.

\bibitem{he2016deep}
K.~He, X.~Zhang, S.~Ren, and J.~Sun, ``Deep residual learning for image
  recognition,'' in \emph{Proceedings of the IEEE conference on computer vision
  and pattern recognition}, 2016, pp. 770--778.

\bibitem{zhang2017poseidon}
H.~Zhang, Z.~Zheng, S.~Xu, W.~Dai, Q.~Ho, X.~Liang, Z.~Hu, J.~Wei, P.~Xie, and
  E.~P. Xing, ``{Poseidon}: an efficient communication architecture for
  distributed deep learning on {GPU} clusters,'' in \emph{Proceedings of the
  2017 USENIX Conference on Usenix Annual Technical Conference}.\hskip 1em plus
  0.5em minus 0.4em\relax USENIX Association, 2017, pp. 181--193.

\bibitem{handley2017re}
M.~Handley, C.~Raiciu, A.~Agache, A.~Voinescu, A.~W. Moore, G.~Antichi, and
  M.~W{\'o}jcik, ``Re-architecting datacenter networks and stacks for low
  latency and high performance,'' in \emph{Proceedings of the Conference of the
  ACM Special Interest Group on Data Communication}.\hskip 1em plus 0.5em minus
  0.4em\relax ACM, 2017, pp. 29--42.

\bibitem{guo2016rdma}
C.~Guo, H.~Wu, Z.~Deng, G.~Soni, J.~Ye, J.~Padhye, and M.~Lipshteyn, ``{RDMA}
  over commodity ethernet at scale,'' in \emph{Proceedings of the 2016 ACM
  SIGCOMM Conference}.\hskip 1em plus 0.5em minus 0.4em\relax ACM, 2016, pp.
  202--215.

\bibitem{you2017scaling}
Y.~You, A.~Bulu{\c{c}}, and J.~Demmel, ``Scaling deep learning on {GPU} and
  {Knights} {Landing} clusters,'' in \emph{Proceedings of the International
  Conference for High Performance Computing, Networking, Storage and
  Analysis}.\hskip 1em plus 0.5em minus 0.4em\relax ACM, 2017, p.~9.

\bibitem{szegedy2015going}
C.~Szegedy, W.~Liu, Y.~Jia, P.~Sermanet, S.~Reed, D.~Anguelov, D.~Erhan,
  V.~Vanhoucke, A.~Rabinovich \emph{et~al.}, ``Going deeper with
  convolutions.''\hskip 1em plus 0.5em minus 0.4em\relax Cvpr, 2015.

\bibitem{rabenseifner2004optimization}
R.~Rabenseifner, ``Optimization of collective reduction operations,'' in
  \emph{International Conference on Computational Science}.\hskip 1em plus
  0.5em minus 0.4em\relax Springer, 2004, pp. 1--9.

\bibitem{thakur2005optimization}
R.~Thakur, R.~Rabenseifner, and W.~Gropp, ``Optimization of collective
  communication operations in {MPICH},'' \emph{The International Journal of
  High Performance Computing Applications}, vol.~19, no.~1, pp. 49--66, 2005.

\bibitem{hoefler2010toward}
T.~Hoefler, W.~Gropp, R.~Thakur, and J.~L. Tr{\"a}ff, ``Toward performance
  models of {MPI} implementations for understanding application scaling
  issues,'' in \emph{European MPI Users' Group Meeting}.\hskip 1em plus 0.5em
  minus 0.4em\relax Springer, 2010, pp. 21--30.

\bibitem{sarvotham2001connection}
S.~Sarvotham, R.~Riedi, and R.~Baraniuk, ``Connection-level analysis and
  modeling of network traffic,'' in \emph{Proceedings of the 1st ACM SIGCOMM
  Workshop on Internet Measurement}.\hskip 1em plus 0.5em minus 0.4em\relax
  ACM, 2001, pp. 99--103.

\bibitem{hang2017paleo}
H.~Qi, E.~R. Sparks, and A.~Talwalkar, ``{Paleo}: A performance model for deep
  neural networks,'' in \emph{International Conference on Learning
  Representations}, 2017, p.~10.

\bibitem{deng2009imagenet}
J.~Deng, W.~Dong, R.~Socher, L.-J. Li, K.~Li, and L.~Fei-Fei, ``{ImageNet}: A
  large-scale hierarchical image database,'' in \emph{Computer Vision and
  Pattern Recognition, 2009. CVPR 2009. IEEE Conference on}.\hskip 1em plus
  0.5em minus 0.4em\relax IEEE, 2009, pp. 248--255.

\bibitem{jia2014caffe}
Y.~Jia, E.~Shelhamer, J.~Donahue, S.~Karayev, J.~Long, R.~Girshick,
  S.~Guadarrama, and T.~Darrell, ``Caffe: Convolutional architecture for fast
  feature embedding,'' \emph{arXiv preprint arXiv:1408.5093}, 2014.

\bibitem{jia2018highly}
X.~Jia, S.~Song, W.~He, Y.~Wang, H.~Rong, F.~Zhou, L.~Xie, Z.~Guo, Y.~Yang,
  L.~Yu \emph{et~al.}, ``Highly scalable deep learning training system with
  mixed-precision: Training {ImageNet} in four minutes,'' \emph{arXiv preprint
  arXiv:1807.11205}, 2018.

\bibitem{li2014communication}
M.~Li, D.~G. Andersen, A.~J. Smola, and K.~Yu, ``Communication efficient
  distributed machine learning with the parameter server,'' in \emph{Advances
  in Neural Information Processing Systems}, 2014, pp. 19--27.

\end{thebibliography}
%\bibliography{cites}
\end{document}